\documentclass{ieeeaccess}
\usepackage{cite}
\usepackage{amsmath,amssymb,amsfonts,amsthm}
\usepackage{algorithm, algorithmic}
\usepackage{graphicx}
\usepackage{textcomp}
\usepackage{subcaption}
\usepackage{mathtools}  
\usepackage{bbm}
\usepackage{soul}
\usepackage{braket}
\mathtoolsset{showonlyrefs}  

\def\BibTeX{{\rm B\kern-.05em{\sc i\kern-.025em b}\kern-.08em
    T\kern-.1667em\lower.7ex\hbox{E}\kern-.125emX}}
\makeatletter
\renewcommand{\fnum@algorithm}{\fname@algorithm}
\makeatother

\newcommand{\real}{\ensuremath{\mathbb{R}}}
\newcommand{\prob}{\ensuremath{\mathbb{P}}}

\newcommand{\naturals}{\ensuremath{\mathbb{N}}}

\newcommand{\subscr}[2]{#1_{\textup{#2}}}

\newcommand{\Norm}[1]{\|#1\|}

\newtheorem{lemma}{Lemma}
\theoremstyle{remark}
\newtheorem*{remark}{Remark}

\newfloat{method}{htbp}{loa}
\floatname{method}{Method}
\makeatletter\newcommand\l@method{\@dottedtocline{1}{1.5em}{2.3em}}\makeatother

\begin{document}

\history{Date of current version Oct 23, 2023.}
\doi{xx.xxxx/ACCESS.2023.xxxxxxx}
\title{Parallel Quantum Rapidly-Exploring Random Trees}

\author{\uppercase{Paul Lathrop}\authorrefmark{1}, \IEEEmembership{Student Member, IEEE},
\uppercase{Beth Boardman}\authorrefmark{2}, and Sonia Mart\'{i}nez\authorrefmark{3},
\IEEEmembership{Fellow, IEEE}}

\address[1]{Department of Mechanical and Aerospace Engineering, University of California, San Diego, La Jolla, CA, 92092 USA, and Los Alamos National Laboratory, Los Alamos, NM 87545 USA (e-mail: pdlathrop@gmail.com)}
\address[2]{Los Alamos National Laboratory, Los Alamos, NM 87545 USA (e-mail: bboardman@lanl.gov)}
\address[3]{Department of Mechanical and Aerospace Engineering, University of California, San Diego, La Jolla, CA, 92092 USA (e-mail:soniamd@ucsd.edu)}
\tfootnote{This work was supported by Los Alamos National Laboratory and is approved
for release under LA-UR-23-31988.}

\markboth
{Lathrop \headeretal: Parallel Quantum Rapidly-Exploring Random Trees}
{Lathrop \headeretal: Parallel Quantum Rapidly-Exploring Random Trees}

\corresp{Corresponding author: Paul Lathrop (e-mail: pdlathrop@gmail.com).}

\begin{abstract}
In this paper, we present the Parallel Quantum Rapidly-Exploring Random Tree (Pq-RRT)
algorithm, a parallel version of the Quantum Rapidly-Exploring
Random Trees (q-RRT) algorithm~\cite{PL-BB-SM:23}. Parallel Quantum RRT is a parallel quantum algorithm formulation of a sampling-based motion planner that uses Quantum Amplitude Amplification to search databases of reachable states for addition to a tree. In this work we investigate how parallel quantum devices can more efficiently search a database, as the quantum measurement process involves the collapse of the superposition to a base state, erasing probability information and therefore the ability to efficiently find multiple solutions. Pq-RRT uses a manager/parallel-quantum-workers formulation, inspired by traditional parallel motion planning, to perform simultaneous quantum searches of a feasible state database. We present results regarding likelihoods of multiple parallel units finding any and all solutions contained with a shared database, with and without reachability errors, allowing efficiency predictions to be made. We offer simulations in dense obstacle environments showing efficiency, density/heatmap, and speed comparisons for Pq-RRT against q-RRT, classical RRT, and classical parallel RRT. We then present Quantum Database Annealing, a database construction strategy for Pq-RRT and q-RRT that uses a temperature construct to define database creation over time for balancing exploration and exploitation.
\end{abstract}
\begin{keywords}
Sampling-Based Motion Planning, Quantum Computing,  Parallel Computing
\end{keywords}
\titlepgskip=-21pt

\maketitle
\section{Introduction}
\IEEEPARstart{Q}{uantum} computing algorithms have shown promise in
accelerating the search for solutions when applied to computationally
intensive, complex problems. More efficient solutions have been found and proven with quantum computing in fields such as pure
science simulations~\cite{FT-AC-SC-DG:20},
cryptography~\cite{SP-UA-LB-MB:20}, and machine
learning~\cite{JB-PW-NP-PR-NW-SL:17}.

With respect to robotics and motion planning, quantum algorithms have also been found to increase speed and efficiency. 
The heart of quantum advantage is derived from quantum parallelism and the ability to perform simultaneous computations on superpositions of states, which motivated our work in~\cite{PL-BB-SM:23}. To aid in sampling-based motion planning, the key efficiency bottleneck is the search for dynamically-reachable, obstacle-free states to connect to the search tree. Unstructured databases of possible next states can be searched simultaneously with Quantum Amplitude Amplification (QAA) to efficiently find an amenable state, but the quantum measurement process forces information loss through the collapse of the superposition. Although all database states are searched in parallel, the process can only return one state. 
In this work we are motivated by efforts in traditional motion planning to parallelize sampling-based planners for efficient path generation using multi-core computers and GPUs, and the ability of quantum algorithms to perform parallel computations. We seek to explore ways for quantum motion planning algorithms to allow multiple, parallel quantum computers to more efficiently search a database of states and return multiple possible solutions. 

\subsection{Literature Review} 
In this section, we provide an overview of quantum computing as it
applies to robotic applications, non-quantum efforts to parallelize
sampling-based motion planning algorithms, the use of annealing and
temperature constructs as they applies to motion planning, and how
this is related to our efforts to increase the efficiency of q-RRT.

Quantum algorithms have been applied to a range of robotic and motion planning-related applications. They have been used to solve
generalized optimization
problems~\cite{MC-AA-RB-SCB-SE-KF-JRM-KM-XY-LC:21}, estimate stochastic
processes~\cite{DSA-CPW:99}, and provide speedup to Monte Carlo
methods~\cite{AM:15}. They have also performed quantum searches~\cite{RP:13} of physical
regions~\cite{SA-AA:03}, found marked elements of a Markov
chain~\cite{FM-AN-JR-MS:11}, and searched more abstract
spaces such as a search engine
network~\cite{ESB-JD-JGG-DZ:12}. A more detailed overview on how quantum computing
has been applied to robotics, along with open questions, can be found at~\cite{CP-MB-HP-MH-BD:19}. 

Quantum computing has also been used to improve motion planning specifically. Quantum reinforcement learning~\cite{KR-QZ-SNB-PF-JRB:21} has increased the speed
and robustness (when compared to classical, non-quantum algorithms) of robotic reinforcement
learning algorithms when learning optimal policies in gridded environments~\cite{DD-CC-LH:07}~\cite{DD-CC-HL-TJT:08}.
An additional use of quantum computing for robotic trajectory planning is addressed in~\cite{LM:15}, which uses a Quantum Evolutionary Algorithm to search for optimal trajectories with a population dynamics/mutation quantum algorithm. Lastly, the review~\cite{NBD-FLP-APA:22} examines quantum control algorithms, and the topic of feedback control accomplished using quantum computing. The work at hand is distinct from the state of the art of quantum computing as applied to motion planners, as we apply such methods to sampling-based planners, which have the advantage of providing fast solutions in high dimensional environments alongside providing probabilistic completeness guarantees~\cite{SML:06}. Besides our previous q-RRT algorithm, quantum computing has not, to the best of our knowledge, been applied to sampling-based planning algorithms. A further overview on how quantum computing has been applied to motion planning and robotics can be found in~\cite{PL-BB-SM:23}.

In the field of non-quantum motion planning, sampling-based planning algorithms
such as Probabilistic Roadmaps (PRMs)~\cite{NA-YW:96} and
Rapidly-exploring Random Trees (RRT)~\cite{JB-LK-JL-TL-RM-PR:96} have
taken the forefront due to their efficiencies in high-dimensional planning spaces
and ability to handle complex robot dynamic
constraints~\cite{ME-MS:14}, such as robotic grasping tasks, autonomous
vehicle planning, and UAV navigation. RRTs and PRMs have been extended in many
ways to improve their sampling speed, exploration ability,
and collision-checking subroutine. RRT* is an important extension regarding path optimality through rewiring~\cite{SK-EF:11}. An overview of sampling-based motion planning can be
found in the textbook~\cite{SML:06}. Three ways to increase
motion planning efficiency are the use of quantum computing, the use
of parallel algorithms and architectures, and the use of sampling
strategies. In this work we consider the combinations of the three
approaches through parallel quantum computing and database
construction strategies, which is akin to sampling strategies in
classical algorithms.

Motion planning algorithms have been written for parallel tree
creation~\cite{MStrandberg:04} and parallel computation with
GPUs~\cite{JB-SK-EF:11}. In~\cite{MStrandberg:04}, local trees are built in parallel to explore difficult regions, and guidelines on when to create and grow local trees are made. In~\cite{JB-SK-EF:11}, the authors parallelize the collision-checking procedure using GPUs to increase
optimal planning speed in high-dimensional spaces.  The
work~\cite{JI-RA:12} outlines Parallel RRT and Parallel RRT*, which
uses lock free parallelism and partition based sampling to provide superlinear
speedup to RRT and RRT*. The work~\cite{DD-TS-JC:13} compares parallel
versions of RRTs on large scale distributed memory architectures,
including \emph{or}-Parallel RRT (multiple simultaneous individual
RRT's) and two methods of collaborative single RRT, Distributed RRT
and Manager-Worker RRT.  The work~\cite{DD-TS-JC:13} also includes a
succinct literature review regarding parallel motion planning and
Parallel RRT. For comparison purposes, in the work at hand we consider
a class of Manager-Worker Parallel RRTs, focusing on the
  parallelization of the collision-checking procedure of RRT.
This is justified by the findings in~\cite{DD-TS-JC:13},
which concludes that for variable expansion cost cases, where the
communication load is insignificant compared to the computation load,
Manager-Worker RRT outperforms, or nearly matches, Distributed RRT in
studied passage, corridor, and roundabout environments.

The work~\cite{RP-AR-AD-SD:21} discusses parallel quantum computing
architectures and control strategies for distributed quantum machines,
noting that multiple few-qubit devices may be more technologically
feasible than single many-qubit devices. In this work, we consider
parallel quantum computers to be multiple simultaneous identical
quantum devices governed by a classical device in order to perform
parallel computation.

A second extension to increase the efficiency of q-RRT relies on
database construction, where we employ a method inspired by simulated
annealing.  Simulated annealing~\cite{SK-CDG-MPV:83} is an
optimization technique that relies on decreasing a
  temperature parameter to find global maxima and minima of a
  nonconvex optimization problem, which is somewhat robust to local
features. Temperature acts as a guide to the probability of
accepting a worse state, allowing an optimizer to explore past local
features while eventually settling on estimates of global
optima when temperature falls. An overview can be
found at~\cite{RWE:90} and~\cite{PJMVL-EHLA:87}. 

Although a
temperature construct is mainly used in the context of optimization,
similar annealing ideas have also been applied to motion planning,
and we intend to apply them to guiding the exploration vs
exploitation trade-off of the planning algorithm. In a manner somewhat reminiscent of annealing, the
work~\cite{AVO-JF-APM:12} uses a dynamic reaction-diffusion process to
propagate, then contract, a search area for a goal location.
Additionally, the covariant Hamiltonian optimization for motion
planning (CHOMP) method~\cite{MZ-NR-ADD-MP-MK-CMD-JAB-SSS:13} uses
gradient techniques to improve trajectories and solve motion planning
queries. CHOMP uses simulated annealing to avoid local minima in
trajectory optimization and not to guide sampling-based motion
planning itself. The work~\cite{SMP-IS:14} uses simulated annealing to
balance exploitation of Sampling-Based A* (SBA*) and exploration of
Rapidly-exploring Random Tree* (RRT*). As cooling occurs, the
probability of choosing the exploration strategy drops and the
probability of choosing the exploitation strategy
increases. Similarly, the transition based
RRT~\cite{LJ-JC-TS:08},~\cite{LJ-JC-TS:10} method uses a temperature
quantity inspired by simulated annealing to define the difficulty
level of transition tests to accept higher cost configurations in an
effort to explore a configuration space. Similar to this work, we use
a temperature quantity to guide the level of exploration, but because
we are using quantum computing with q-RRT to perform motion planning,
temperature factors into database construction rather than individual
samples themselves.

In~\cite{PL-BB-SM:23}, we introduced how quantum computing methods can
be applied to sampling-based motion planning in two ways, a full path
database search and an RRT-based single-state database search
q-RRT. The Quantum Rapidly-Exploring Random Trees algorithm, q-RRT,
uses Quantum Amplitude Amplification (QAA) to search databases of
possible reachable states. A focus of our work~\cite{PL-BB-SM:23} was
in estimating solution likelihood, so QAA could be performed an
optimal number of times, through the use of random square lattice
environments and numerical simulations. In the work at hand, we shift
focus to address a particular shortcoming of quantum computers and
qubits: this approach suffers from the limits of quantum mechanics, as
qubits cannot be copied and the quantum measurement process admits a
single solution. We study how quantum devices, working
in parallel, can efficiently solve motion planning problems, while
generalizing environments away from random square lattices. Instead of
focusing on how many solutions exist within a database, we focus on
how multiple solutions can efficiently be found from a single
database.

\subsection{Contributions}
The main contributions of this work are the following. 

\begin{itemize}
\item Creation of a parallel quantum computing variation to q-RRT, called Parallel q-RRT (Pq-RRT), which uses a parallel quantum computing structure to allow multiple solutions from a single database in general obstacle environments;
\item Characterization of key probability values for multiple quantum workers searching a shared database, with and without false positive and false negative oracle errors in order to minimize efficiency loss;
\item Creation of a construction strategy for quantum-search databases, called Quantum Database Annealing, which uses a temperature construct to select sample distances and balance exploration vs exploitation;
\item Demonstration (through simulation) of the increased efficiency of Pq-RRT over q-RRT, as compared to the efficiency increase of Parallel RRT over RRT;
\item Simulations of faster tree exploration with Quantum Database Annealing as compared to standard uniform-sampling database construction.
\end{itemize}

\section{Organization and Notation}
In Section~\ref{section:qRRT}, for reference purposes we provide a working definition of q-RRT from~\cite{PL-BB-SM:23}. In Sections~\ref{section:pqRRT} and~\ref{section:pqRRTprob}, we define Parallel q-RRT then key probability results for Pq-RRT, respectively. In Section~\ref{section:annealing}, we define the database construction strategy Quantum Database Annealing. In Section~\ref{section:resultsmain}, we provide runtime and efficiency results for q-RRT and Pq-RRT as compared to RRT and Parallel RRT. In section~\ref{section:resultsspeed}, we provide heatmaps of q-RRT's node placement speed over RRT, and in Section~\ref{section:resultsthinchannel} we provide narrow corridor results for q-RRT. In Section~\ref{section:resultsannealing}, we provide tree comparisons between Quantum Database Annealing and standard database construction.
\subsection{Notation}
The set of natural numbers is denoted as $\naturals$,
and similarly, $\real$ is the set of real numbers, while $\mathbb{C}$
is the set of complex numbers. For $d,p \in \naturals$, we use $\real^d$ to refer to
$d$-dimensional real vector space, and by $x\in \real^d$ a vector in
it. Lastly, the Euclidean norm in $\real^d$ is identified with
$\Norm{\cdot}_2$. 
\subsection{Quantum Computing Basics}
We briefly introduce quantum computing basics and how we use
quantum algorithms to solve motion planning problems. We defer a
 more in depth introduction of quantum computing basics
 to~\cite{PL-BB-SM:23} and a comprehensive circuit-level discussion
 to~\cite{JPreskill:98}. Let $\ket{z}$
refer to the quantum state represented by the qubit $z$. Quantum
computers encode information in basic units called qubits, given as
the superposition of two basis quantum states, $\ket{0}$ and
$\ket{1}$. The set $\{\ket{0},\ket{1}\}$ defines a basis of
quantum states. Qubits maintain probability amplitudes (relative
measurement likelihoods) $\alpha$ and $\beta$. A qubit $\ket{\Psi}$
can exist in a superposition of $\ket{0}$ and $\ket{1}$, of the form
$\ket{\Psi} = \alpha \ket{0}+\beta\ket{1}$ with
$\alpha,\beta \in \mathbb{C}$, $\;|\alpha|^2+|\beta|^2=1$. Multiple
qubits come together to form a memory storage unit called a qubit
register.  The measurement process involves the collapse of the state
$\ket{\Psi}$ to a base state $\{\ket{0},\ket{1}\}$ according to probabilities
$\alpha^2$ and $\beta^2$ (known as the Born rule).

Quantum algorithms perform fast parallel computations on
superpositions in a process known as quantum
parallelism~\cite{DD-CC-HL-TJT:08}. This manipulates the amplitudes
$\alpha$ and $\beta$ of the system, which cannot be known
explicitly, as any measurement collapses the qubit. 
Previously we leveraged quantum parallelism to create a new algorithm for solving motion planning problems.

Our Quantum Rapidly-Exploring Random Tree (q-RRT)
algorithm~\cite{PL-BB-SM:23} uses Quantum Amplitude Amplification
(QAA) with a Boolean oracle function $\mathcal{X}$ that evaluates
reachability of states. QAA uses an oracle $\mathcal{X}$ to increase the
probability of measuring a selected (`good') state $\Psi$ (such that $\mathcal{X}(\Psi) = 1$). The QAA precise definition, mechanism of action,
and discussion can be found at~\cite{GB-PH-MM-AT:02}, page 56. Similar
to~\cite{PL-BB-SM:23}, we take advantage of QAA to quantum search a
size-$N$ unordered database for oracle-tagged items in
$\mathcal{O}(N^{1/2})$ oracle calls, whereas non-quantum search
algorithms require $\mathcal{O}(N)$ calls. 
\section{Parallel Quantum RRT and Quantum Database Annealing}
In Section~\ref{section:qRRT}, we outline the q-RRT Algorithm
from~\cite{PL-BB-SM:23}, followed by a presentation of
the Parallel Quantum RRT Algorithm in Section~\ref{section:pqRRT}, and then probability results for
Pq-RRT in Section~\ref{section:pqRRTprob}. Lastly, in Section~\ref{section:annealing}, we present the
Quantum Database Annealing strategy.

\subsection{Quantum RRT Algorithm}\label{section:qRRT}
The q-RRT algorithm from our previous work~\cite{PL-BB-SM:23} is a tree-based search algorithm based on
RRTs~\cite{JJK-SML:00}.
Quantum RRT uses QAA on a database of
possible parent-child pairs to admit reachable points to the tree. In
this work, q-RRT returns a path (as opposed to the full tree
in~\cite{PL-BB-SM:23}) and has the end condition of finding a goal.
The line-by-line algorithm can be found at\cite{PL-BB-SM:23}, and because of q-RRT's similarities with Pq-RRT Manager (Alg.~\ref{alg:pqmanager}) and Pq-RRT Worker (Alg.~\ref{alg:pqworker}), we omit the line-by-line and instead reference lines of the latter two enumerated algorithm descriptions.

The q-RRT Algorithm takes as input an initial state $x_0$ and a goal
state $\subscr{x}{G}$ in a compact configuration space
$C \subseteq \real^d$, a number of qubit registers $n$, and a quantum
oracle function $\mathcal{X}$. It returns a dynamically feasible
obstacle free path $\gamma$. The q-RRT algorithm adds nodes to graph
$T$ until there is a node within distance $\delta$ of the goal
$\subscr{x}{G}$. To add a node, q-RRT creates a $2^n$ sized database
$D$ of random possible nodes and the nearest parent in $T$ to the
random node, as shown in Alg.~\ref{alg:pqmanager} on
lines~\ref{alg2:linedatabase1}-\ref{alg2:linedatabase2}. A
$1-$to$-1$ mapping $F$ is created between database $D$ and qubit
$\ket{\Psi}$ (shown in Alg.~\ref{alg:pqworker}, line~\ref{algworker:linemapping}). Then, the qubit is
initialized and an equal superposition between states set (shown in Alg.~\ref{alg:pqworker} on
lines~\ref{algworker:lineinitqubit} and~\ref{algworker:linesuperposition} respectively). Let $\mathbf{W}$ be the Walsh-Hadamard transform, the
operator which maps a qubit to an equal superposition of all qubit
states. On lines~\ref{algworker:lineQAAstart}-\ref{algworker:lineQAAend} of Alg.~\ref{alg:pqworker}, the operator $Q$ performs QAA to amplify the probability amplitudes of
correct states as defined by oracle $\mathcal{X}$, which tests the
reachability of random samples to the nearest proposed parent $P$ of
the existing graph $T$, thus ensuring that $T$ is fully reachable. For an analysis into selecting $\subscr{i}{max}$ we refer readers to~\cite{PL-BB-SM:23}.

Measurement is performed and the correct database element selected in Alg.~\ref{alg:pqworker} on
line~\ref{algworker:linemeasure}. After measurement is performed, the quantum state has collapsed and no further information (beyond the selected state) can be gained from the qubit. 
Lines~\ref{algmanager:linestartgoal}-\ref{algmanager:lineendgoal} of Alg.~\ref{alg:pqmanager} allow
a node placed within $\delta$ of $\subscr{x}{G}$ to be admitted to $T$
as $\subscr{x}{G}$, ending the algorithm. The path is returned after
successful loop execution on Alg.~\ref{alg:pqmanager}, line~\ref{algmanager:linereturn}.

\subsection{Parallel Quantum RRT}\label{section:pqRRT}
In this section, we define the Parallel Quantum RRT (Pq-RRT) algorithm
as a manager (Alg.~\ref{alg:pqmanager}) worker
(Alg.~\ref{alg:pqworker}) formulation. The Pq-RRT algorithm performs
reachability tests using a parallel pool of quantum computers, and is a direct extension of q-RRT inspired by parallel motion planning. The
manager algorithm assigns work to the parallel pool and adds results
to the tree $T$. The assigned work consists of each quantum worker performing a reachability check on a
database $D$ using QAA with a quantum oracle, and returning a single database element.
The specific parallelization architecture is chosen for a few reasons. We consider scenarios where generally worker runtime cost dominates the message passing cost (as per~\cite{DD-TS-JC:13}). This rules out such architectures as disjoint workers independently searching for a solution by growing separate trees, which have relatively little message passing but are much less runtime-efficient in finding a solution. 

In the chosen manager-worker scheme, instead of discretizing a search space to allow workers to each grow a separate part of a tree, each worker is tasked with adding a single element to the tree (anywhere). This removes the idleness aspect of workers, as workers do not have to be actively listening for tree updates and do not rely on the work of others to perform their own search. Additionally, because of the probabilistic nature of the quantum search process, the parallel quantum routine generally can have all workers complete work simultaneously. This feature is not possible in non-quantum parallel architectures, as each worker is performing a stochastic search for a solution, which generally takes differing times between workers. In the quantum architecture, however, the runtime to amplify a database is much more consistent, and if each quantum worker is performing the same number of amplifications before measurement, they should complete a search nearly simultaneously. The key difference lies in the goal of the work performed. A solution does not need to be deterministically found for work to be completed (as in the non-quantum case). Work is instead completed when a solution is more likely to be found, which can be standardized for runtime across the workers. 

Alg.~\ref{alg:pqmanager}, the manager algorithm, has the same
inputs and outputs as q-RRT. The worker
algorithm, Alg.~\ref{alg:pqworker}, admits as inputs the current tree
$T$, the number of qubit registers $n$, the quantum oracle function
$\mathcal{X}$, and a copy of the shared database $D$, and returns a
selected element of the database $[\subscr{x}{add},P]$. 

\begin{algorithm}
\caption{$\mathbf{1}$ Pq-RRT Manager, shared database}
\begin{algorithmic}[1] \label{alg:pqmanager}
\renewcommand{\algorithmicrequire}{\textbf{Input:}}
\renewcommand{\algorithmicensure}{\textbf{Output:}}
\REQUIRE $x_0,\subscr{x}{G}\;n,\;\text{oracle } \mathcal{X}$
\ENSURE Path $\gamma$
\STATE Init $p$-worker pool
\STATE Init tree $T$ with root at $x_0$
\WHILE{$\subscr{x}{G}\notin T$} \label{alg2:linedatabaseplace}
\FOR{$i = 1 \text{ to } 2^n$} \label{alg2:linedatabase1}
\STATE $t =$ random point
\STATE $P =$ closest parent of $t$ in $T$
\STATE $D(i) = [t;\:P]$
\ENDFOR \label{alg2:linedatabase2}
\STATE $[\subscr{x}{add},P](k) =\text{ Worker}(T,n,\mathcal{X},D), \text{ for } k\in p$
\FOR{$k = 1 \text{ to } p$}
\IF{$[\subscr{x}{add},P](k) \notin T$ }
\STATE Add $[\subscr{x}{add},P](k)$ to $T$ 
\ENDIF
    \IF{$\|\subscr{x}{add}(k)-\subscr{x}{G}\|<\delta$} \label{algmanager:linestartgoal}
        \STATE $\subscr{x}{add}(k) = \subscr{x}{G}$
    \ENDIF \label{algmanager:lineendgoal}
\ENDFOR
\ENDWHILE
\STATE Return path $\gamma$ from $T$ \label{algmanager:linereturn}
\end{algorithmic}
\end{algorithm}

\begin{algorithm}
\caption{$\mathbf{2}$ Pq-RRT Worker, shared database}
\begin{algorithmic}[1] \label{alg:pqworker}
\renewcommand{\algorithmicrequire}{\textbf{Input:}}
\renewcommand{\algorithmicensure}{\textbf{Output:}}
\REQUIRE $T,n,\mathcal{X},D$
\ENSURE $[\subscr{x}{add},P]$
\STATE Enumerate $D$ via $F:\{0,1\}^n \; \rightarrow D$ \label{algworker:linemapping}
\STATE Init $n$ qubit register $\ket{z} \gets \ket{0}^{\otimes n}$ \label{algworker:lineinitqubit}
\STATE $\ket{\Psi} \gets \mathbf{W}\ket{z}$ \label{algworker:linesuperposition}
\FOR{$i = 1 \text{ to } \subscr{i}{max}$} \label{algworker:lineQAAstart}
\STATE $\ket{\Psi} \gets Q(\mathcal{X})\ket{\Psi}$
\ENDFOR \label{algworker:lineQAAend}
\STATE $[\subscr{x}{add},P] \gets F($measure$(\ket{\Psi}))$ \label{algworker:linemeasure}
\STATE Return $[\subscr{x}{add},P]$ 
\end{algorithmic}
\end{algorithm}

Two versions of this parallel formulation are possible, shared and
unshared database. The fundamental difference is whether parallel pool
workers create a database or perform QAA on copies of the same
database $D$, which is created by the manager. For the shared database
version, as shown in Fig.~\ref{fig:diagramshared}, in
Alg.~\ref{alg:pqmanager}
Lines~\ref{alg2:linedatabase1}-\ref{alg2:linedatabase2} the manager
creates the database $D$ and passes copies to each worker $k\in p$, as
shown by the inputs to Alg.~\ref{alg:pqworker}. In this way, the
workers would ``share'' and search (copies of) the same
database. The manager ignores additional identical solutions returned by different workers, which is a fast process given that the workers essentially are returning an index to a database element. 

For an unshared database, as shown in Fig.~\ref{fig:diagramunshared},
the database construction step is performed within the worker
algorithm (Alg.~\ref{alg:pqworker}), which can be a classical worker until
QAA is performed. 

The advantage of the shared database approach is an 
increase in database-use-efficiency due to extracting multiple possible
reachable states per database construction. This aligns with the main motivation behind this work, which is to mitigate the probability information loss due to quantum measurement collapse. Because quantum computers are reducing the time spent on the computationally intensive portion of the algorithm (state collision/reachability checks), steps such as database construction will become a larger proportion of algorithm runtime, so it is advantageous to have high database-use-efficiency. We reserve additional runtime analysis for future work.
However, the shared
database approach can be slightly less oracle-call efficient (compared to unshared database), with
fewer reachable states are admitted per oracle call because repeated
identical solutions are discarded. This is shown in
Section~\ref{section:resultsmain} Fig.~\ref{fig:nodespeedsimulated}, and we discuss how important that efficiency loss is and how to mitigate it in the next section. 

\begin{figure}[h]
	\centering
	\includegraphics[width=.45\textwidth]{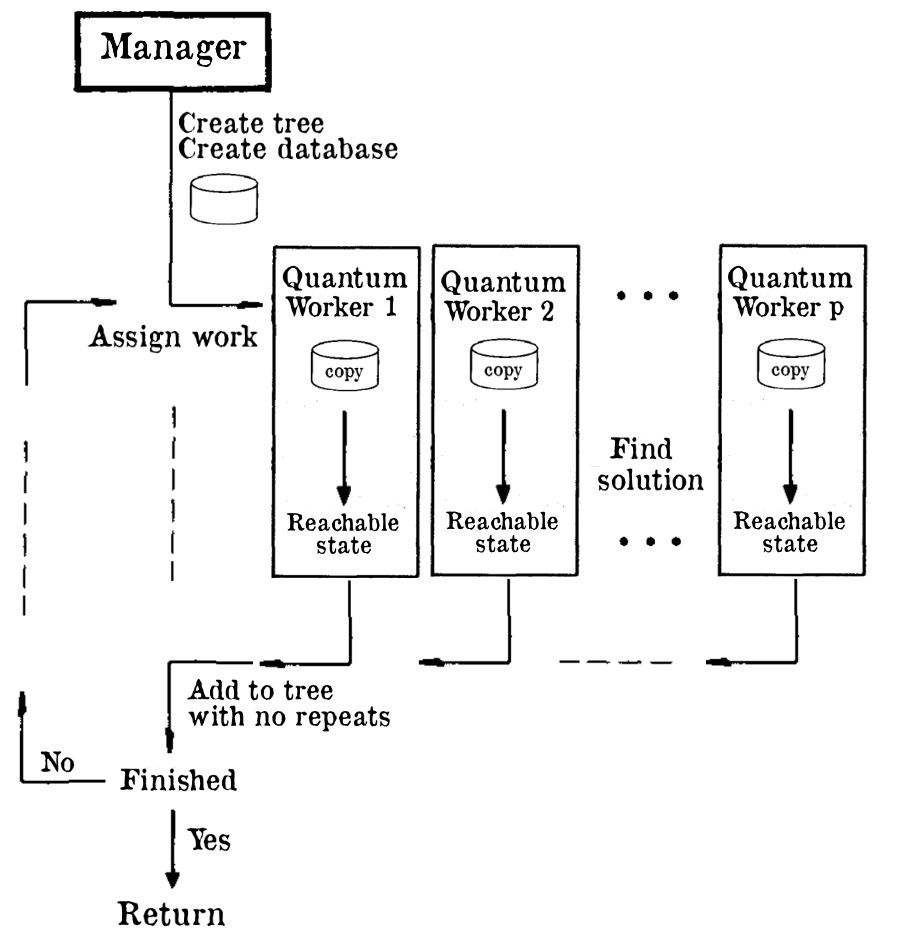}
        \caption{A graphical depiction of the shared database Quantum Parallel RRT algorithm. The manager (Alg.~\ref{alg:pqmanager}) creates a database and passes copies to the quantum workers (Alg~\ref{alg:pqworker}).}
	\label{fig:diagramshared}
\end{figure}

\begin{figure}[h]
	\centering
	\includegraphics[width=.45\textwidth]{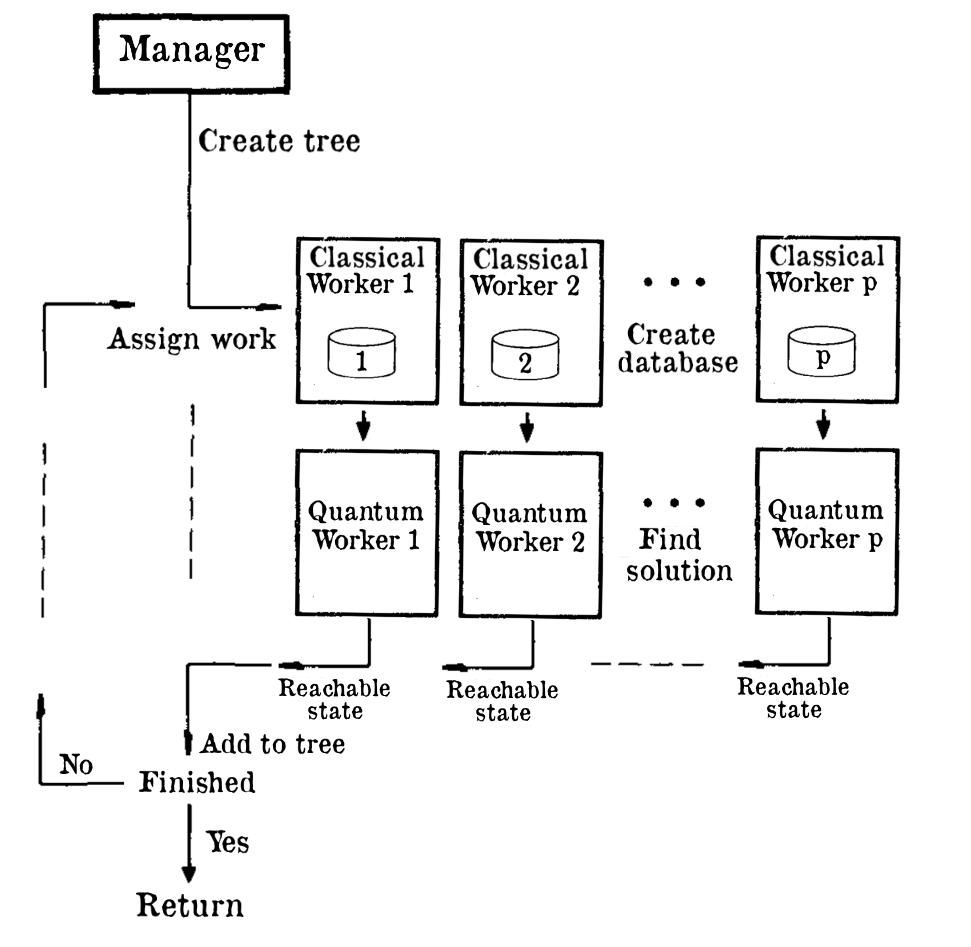}
	\caption{A graphical depiction of the unshared database Quantum Parallel RRT algorithm. The manager prompts $p$ classical workers to create $p$ different databases, which are passed to $p$ quantum workers to find solutions, which are returned to the manager.}
	\label{fig:diagramunshared}
\end{figure}
\subsection{Pq-RRT Probability Results}\label{section:pqRRTprob}
In this section, we characterize $p$
parallel quantum workers finding multiple solutions when Pq-RRT is
operating in a shared database setup. Since all
the workers are independently analyzing the same $2^n$-sized database
$D$, with $m$ oracle-marked solutions, in general multiple workers may
arrive at the same solution. This represents an efficiency loss to the shared database setup, so in what follows we characterize the worst and best case events. The worst case event is all workers arriving at the same solution, which has the runtime performance as non-parallel q-RRT but is $p$-times less oracle call efficient. The best case event is each worker finding a different solution, which has no runtime or oracle call efficiency loss, and ends with $p$ solutions. Only when $p\geq m$ can all solutions be
found in a single parallel pass. To build fast solutions, an understanding of the effects of choices of $p$ (and to some extent $m$) is necessary to maximize efficiency.

We assume that the database is optimally
amplified according to $\subscr{i}{max}$ applications of $Q$, where
$\subscr{i}{max}$ is given by,
\begin{equation}\label{eq:imax}
    \subscr{i}{max}=\frac{\pi}{4}\sqrt{2^n/m} \;;
\end{equation}
see~\cite{CF:03}. We
note that our connectivity analysis on estimating database correctness
in~\cite{PL-BB-SM:23} can be applied to the section at hand in order to attain optimal amplification (to maximize chances of measuring a solution). In what follows, let $G$ be the event that a good state, as defined by the oracle, is measured by a worker after $\subscr{i}{max}$ iterations of $Q$. 
\subsubsection{Without Oracle Errors}
\begin{lemma}\label{lemma:1}
 Let there be a parallel quantum process with $p$ workers and a shared
  $2^n$-sized database with $m$ solutions. The probability that all $p$ workers find the same solution
  is,
\begin{equation}\label{eq:allsame}
    \prob(\text{same solution}) = \prob(G)^p m^{1-p},
\end{equation}
where $\prob(G)$ is the total probability of event $G$,
\begin{equation}
  \prob(G) = \sin^2((2\subscr{i}{max}+1)\theta),
\end{equation}
and where $\theta$ is defined such that
$\sin^2(\theta) = \frac{m}{2^n}$.
\end{lemma}
\begin{proof}
  To attain this result, we observe that after $\subscr{i}{max}$
  iterations of $Q$, all $m$ solutions have equal probability of
  measurement given by $\prob(G)/m$. The probability that all $p$
  workers measure a \emph{particular} solution $i$ is,
\begin{equation}
  \prob(\text{particular solution}) = \left(\frac{\prob(G)}{m}\right)^p,
\end{equation}
and this is multiplied by $m$ to generalize to finding \emph{any} same
solution, yielding Eq.~\eqref{eq:allsame}. The total good measurement probability and the definition of $\theta$ can be found at~\cite{GB-PH-MM-AT:02}.
\end{proof}

\begin{lemma} \label{lemma:2} For $m \geq p$ and $m,p \in \naturals$, the probability that all
  workers find different solutions is given by,
  \begin{equation}\label{eq:alldiff}
    \prob(\text{different solutions}) = \frac{\prob(G)^p m!}{m^p(m-p)!},
  \end{equation}
\end{lemma}
\begin{proof}
  This result follows from $m$ permute $p$ (the number of possible
  ways $p$ objects can be selected, without replacement, from $m$
  possibilities) over the total number of possible outcomes
  $m^p$. This is scaled by the likelihood that all workers find a
  correct solution, $\prob(G)^p$, to yield Eq.~\eqref{eq:alldiff}.
\end{proof}

For the worst case scenario in Eq.~\ref{eq:allsame}, as $m \rightarrow \infty$, $\prob(\text{same solution})\rightarrow 0$. This makes intuitive sense: as the number of available solutions increases, the likelihood of all workers finding the same solution decreases as a function with power $1-p$, where $p\geq 2$, as it is only sensible to conjecture about the parallel behavior of $2$ or more workers.

For the best case scenario in Eq.~\ref{eq:alldiff}, as $m \rightarrow \infty$, $\prob(\text{different solutions}) \rightarrow \prob(G)^p$. This also makes intuitive sense: as the number of solutions increases, the likelihood of all workers finding different solutions approaches the total likelihood of all workers finding a solution. We note that for $m \geq p$ and $m,p \in \mathbb{N}$:
\begin{equation}
    \lim_{m\rightarrow \infty} \frac{m!}{m^p(m-p)!} = 1.
\end{equation}
The number of solutions $m$ should be as large as practicable to reduce efficiency loss through oracle overlap.

\begin{lemma}\label{lemma:3}
For $p \geq m$, the expected number of workers $p$, to find all $m$ solutions within $D$ in one pass, is given by,
    \begin{equation}
      E(p) = \frac{m H_m}{\prob(G)},
    \end{equation}
    where $H_m$ is the $m$\textsuperscript{th} harmonic number. Equivalently, $\frac{E(p)}{p_2}$ also describes the expected number of passes at a single database that the set number of workers $p_2$ must make to find all solutions.
\end{lemma}
\begin{proof}
We reach this result from the application of the Coupon Collector's
problem~\cite{PN:08}, with a minor modification, to the independent
quantum computing worker processes. Briefly, the coupon collector's problem concerns questions about the ``collect all coupons from cereal boxes to win'' contest. In this context, solutions in the database represent coupons (to be found with a certain probability), and number of workers represents (the expectation of) how many cereal boxes must be opened to find one of each solution/coupon. The result takes into account that workers may return the same solution. This application is scaled by the total probability of correct solutions, $\prob(G)$, to account for the proportion of the time when a good solution is not measured.
\end{proof}
Lemma~\ref{lemma:3} allows a parallel (and repeated) architecture to be chosen based upon knowledge of $m$, such as from our connectivity analysis on estimating database correctness in~\cite{PL-BB-SM:23}. Additionally, this leads to the database construction tool described in the following section that allows the proportion $m/2^n$ to be made larger or smaller. We remark that in general, it is possible to calculate the probability of $n$ workers coinciding on the same exact solution, as it relates to the multinomial distribution. 

\subsubsection{With Oracle Errors}
We also consider the case where the oracle is making repeated false
positive and false negative errors. Let the probability that a state
is marked by the oracle incorrectly as good be given by $q\in[0,1]$
(false positive), and let the probability that a state is marked by
the oracle incorrectly as bad be given by $v\in[0,1]$ (false
negative). Let the number of ground-truth solutions in the database
tagged by the oracle as a solution be $m_1\leq m$. Let the number of
actual ground truth solutions in the database, which were mistakenly
tagged by the oracle as bad be $m_2$, such that,
\begin{equation}
    q = \frac{1-m_1}{m},\;v = \frac{m_2}{2^n-m},
\end{equation}
as shown in Fig.~\ref{fig:databaseproportions}. 

First, we derive the likelihood that a single worker finds a real solution.
Let $G^*$ be the event that a real, ground truth solution (not according to the oracle) is measured for addition to the tree after optimal amplification with QAA. 
\begin{lemma}
    The total probability of measuring real, ground truth solutions is,
    \begin{equation}
        \prob(G^*) = \frac{m_1}{m}\prob(G)+\frac{m_2}{2^n-m}(1-\prob(G)).
    \end{equation}
\end{lemma}
\begin{proof}
We attain this result by adding the probability of measuring a correctly tagged good solution, $\frac{m_1}{m}\prob(G)$ to the probability of measuring an incorrectly tagged bad solution, $\frac{m_2}{2^n-m}(1-\prob(G))$.
\end{proof}

\begin{figure}[h]
	\centering
	\includegraphics[width=.47\textwidth]{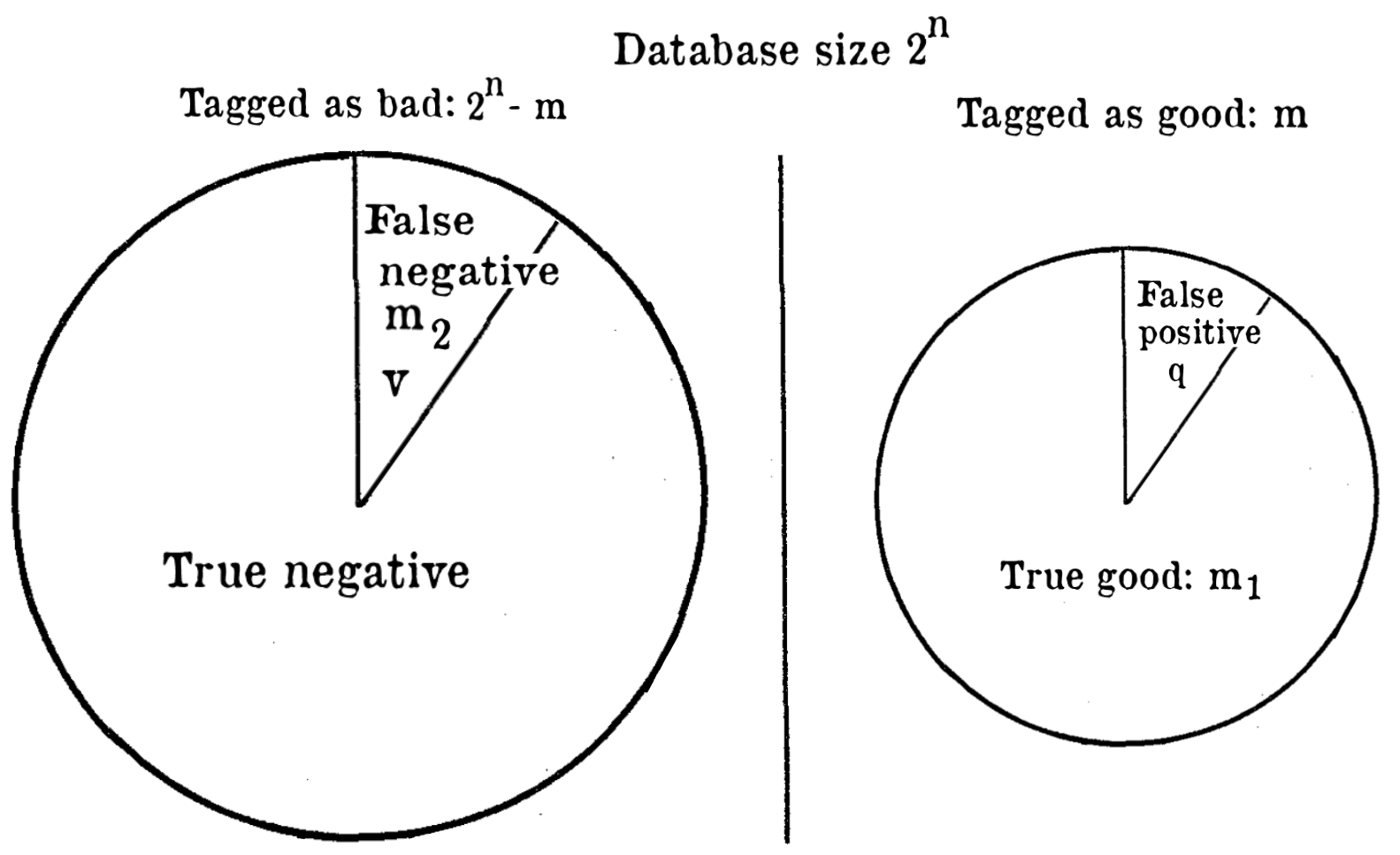}
	\caption{A graphical depiction of the false positive and false negative regions of a database with good and bad tags by an oracle. Of the $m$ good tags, $m_1$ are true good tags, with a false positive probability of $q$. Of the $2^n-m$ bad tags, $m_2$ are actually good elements, with a false negative probability of $v$. After optimal QAA, the probability of a tagged-as-good state being measured is $\prob(G)$.}
	\label{fig:databaseproportions}
\end{figure}

The following lemma is a modification of Lemma~\ref{lemma:1} to include oracle errors.
\begin{lemma}\label{lemma:5}
    The probability that all $p$ workers find the same ground truth solution, adjusted for oracle mistakes, is,
    \begin{equation}\label{eq:allsamemod}
        \prob(\text{same solution}) = m_1\left(\frac{\prob(G)}{m}\right)^p + m_2\left(\frac{1-\prob(G)}{2^n-m}\right)^p.
    \end{equation}
\end{lemma}
\begin{proof}
This result follows from the addition of the probability of all $p$ workers measuring a \emph{particular} correctly tagged good solution, $\left(\frac{\prob(G)}{m}\right)^p$, and the probability of all $p$ workers measuring a \emph{particular} incorrectly tagged bad solution, $\left(\frac{1-\prob(G)}{2^n-m}\right)^p$. Each of these is multiplied by $m_1$ and $m_2$ respectively, to consider \emph{any} real solution, then added together to yield Eq.~\eqref{eq:allsamemod}.
\end{proof}
For the worst case scenario in Lemma~\ref{lemma:5}, again, as $m \rightarrow \infty$, $\prob(\text{same solution}) \rightarrow 0$. 
The following lemma is a modification of Lemma~\ref{lemma:3}, adjusted for oracle false positive errors.
\begin{lemma}\label{lemma:6}
    The expected number of workers $p$ to find all $m_1$ ground truth solutions, when false positive oracle errors are taken into account, is given by,
    \begin{equation}
        E(p^*) = \frac{m_1H_{m_1}}{\frac{m_1}{m}\prob(G)},
    \end{equation}
    where $H_{m_1}$ is the $m_1$\textsuperscript{th} harmonic number.
\end{lemma}
\begin{proof}
This result follows similarly to Lemma~\ref{lemma:3}, with $m_1$ rather than $m$, and where the expectation from the Coupon Collector's Problem is scaled by the proportion of the time that a ground truth good solution is found, $\frac{m_1}{m}\prob(G)$. 
\end{proof}
\begin{remark}
Oracle false negative errors, when considered in Lemma~\ref{lemma:6}, transform the problem into a simple version of the Weighted Coupon Collector's Problem~\cite{PB-TS:09} (also known as McDonald's Monopoly), with some ``rare coupons'' to find (represented by the  $m_2$ false negative solutions), and some ``common coupons'' (represented by the $m_1$ real positive solutions).
In reality, after optimal amplification, the probability of measuring the $m_2$ good solutions incorrectly tagged as bad, $\frac{m_2}{2^n-m}(1-\prob(G))$, is small, and additionally we care more about finding correctly identified good solutions than determining incorrectly tagged bad solutions from a database.
\end{remark}

\subsection{Quantum Database Annealing}\label{section:annealing}
In this section, we define the Quantum Database Annealing (QDA) strategy, shown in Alg.~\ref{alg:QDA}. QDA builds databases with elements constrained to a certain distance from the parent node, as defined by a temperature matrix $H$ and iterator $h$. The QDA strategy is an alternative to standard database construction and is inspired by the optimization technique simulated annealing and our investigations into oracle call constraints in~\cite{PL-BB-SM:23}. It represents a possible way to guide database construction to achieve a particular algorithmic goal, such as approximately selecting $m$ (with regard to the previous section), or in this case, initial fast expansion via spread node placement followed by increasing density through closer node placement. 

In a broad sense, sampling strategies for motion planning have been explored since the beginning, with strategies such as medial axis sampling, boundary sampling, Gaussian (obstacle) sampling, goal biasing, and hybrid schemes~\cite{ME-MS:14}. QDA is distinct from current classical computing approaches because the initial goal in this quantum formulation is to make sample connections less likely. When paired with a large database, QDA exploits quantum computing's ability to quickly find unlikely solutions, resulting in a motion planner that can explore very quickly when measured on oracle calls.

QDA first samples according to a uniform distribution over the configuration space $C$. The nearest (Euclidean distance) existing node to the sample is chosen as the parent. Next, the resulting parent-child relationship is constrained to be within a ball of radius $H(h)$ (with iterator $h$) while maintaining child-sample direction. The resultant pair is added to the database. An alternative is to sample initially over a disc or boundary at a distance constrained by $H(h)$.

In the beginning, with high temperature (when $H(h)$ is large), QDA will build a large database of further away and therefore less likely solutions. This allows further reachable solutions to be found quickly as compared to q-RRT and RRT. As the path planning problem continues (as $h$ increases), the temperature ($H(h)$) may drop to account for the addition of new nodes and to increase the ratio of good solutions in the database. The database size $2^n$ may also drop throughout the problem to increase efficiency, as when there are more solutions, smaller databases function as well as larger. When no additional information is known to guide sampling region, an alternative but similar sampling method in very large bounded configuration spaces is to build extremely large databases of unlikely solutions in an attempt to span long obstacle free channels quickly.

\begin{algorithm}
\caption{$\mathbf{4}$ q-RRT with Quantum Database Annealing}
\begin{algorithmic}[1] \label{alg:QDA}
\renewcommand{\algorithmicrequire}{\textbf{Input:}}
\renewcommand{\algorithmicensure}{\textbf{Output:}}
\REQUIRE $x_0,\subscr{x}{G}\;n,\;\text{oracle } \mathcal{X}$
\ENSURE Path $\gamma$
\STATE Init tree $T$ with root at $x_0$
\STATE \underline{Define temperature array $H$, index $h = 0$}\label{line:alg4:temparray}
\WHILE{$\subscr{x}{G}\notin T$}
\FOR{$i = 1 \text{ to } 2^n$} 
\STATE $t =$ random point
\STATE $P =$ closest parent of $t$ in $T$
\STATE \underline{Constrain $t$ to disc of dist. $H(h)$ from $P$} \label{line:alg4:tempconstraint}
\STATE $D(i) = [t;\:P]$
\ENDFOR 
\STATE Enumerate $D$ via $F:\{0,1\}^n \; \rightarrow D$ 
\STATE Init $n$ qubit register $\ket{z} \gets \ket{0}^{\otimes n}$
\STATE $\ket{\Psi} \gets \mathbf{W}\ket{z}$
\FOR{$i = 1 \text{ to } 2$} 
\STATE $\ket{\Psi} \gets Q(\mathcal{X})\ket{\Psi}$
\ENDFOR 
\STATE $[\subscr{x}{add},P] \gets F($measure$(\ket{\Psi}))$
\IF{$\|\subscr{x}{add}-\subscr{x}{G}\|<\delta$}
\STATE $\subscr{x}{add} = \subscr{x}{G}$
\ENDIF
\STATE Add $[\subscr{x}{add},P]$ to $T$
\STATE \underline{$h++$} \label{line:alg4:iterator}
\ENDWHILE
\STATE Return path $\gamma$ from $T$

\end{algorithmic}
\end{algorithm}

The differences between q-RRT and q-RRT with QDA, Alg.~\ref{alg:QDA}, lie in the latter algorithm's lines~\ref{line:alg4:temparray},~\ref{line:alg4:tempconstraint}, and~\ref{line:alg4:iterator}. Alg.~\ref{alg:QDA} line~\ref{line:alg4:temparray} is where the temperature array $H$ is defined and iterator $h$ initialized. On Alg.~\ref{alg:QDA} line~\ref{line:alg4:tempconstraint}, the temperature constraint is carried out by modifying the random point $t$ with respect to $P$, the closest parent of $t$ in $T$. On Alg.~\ref{alg:QDA} line~\ref{line:alg4:iterator}, the iterator $h$ is incremented to allow different temperatures on future database constructions. In the defined formalism, the database size $2^n$ is set and not decreased.

\section{Results and Discussion}
In this section, we show tree creation comparison results within two dimensional obstacle environments for Pq-RRT, q-RRT, Parallel RRT, and RRT. Direct comparisons highlighting the simulated quadratic runtime advantage of q-RRT over standard RRT are shown in detail in our previous work at~\cite{PL-BB-SM:23}.
Unless otherwise stated, results are presented comparing algorithm performance for solving the same problem in the same randomized obstacle environments. Both Pq-RRT and Parallel RRT are implemented with eight cores (workers). Both quantum algorithms use databases of size $2^8$, and the classical versions of the algorithms (RRT and Parallel RRT) replace the database construction and quantum search process with single reachability tests. The specific version of Parallel RRT is a Manager-Worker formulation (outlined in~\cite{DD-TS-JC:13} under Manager-Worker RRT), where a manager processes the tree and assigns single-node expansion work to workers, as expansion is the computationally expensive part of planning. 

All path planning simulations are run with Matlab v2022b on an eight core MacBook Pro with M2 chip. Quantum states and algorithms are simulated with the Matlab Quantum Computing Functions library~\cite{CF:03}.
All algorithms use the following arbitrary dynamics and reference tracking controller to test reachability for node admittance to the tree,
\begin{equation*} 
  x(t+1) = 
  Ax(t)+Bu(t)\:,
  \: x(0) = \subscr{x}{parent},
\end{equation*}
\begin{equation*}
  A = \begin{bmatrix}-1.5&-2\\
    1&3\end{bmatrix}, \,
  B= \begin{bmatrix}0.5&0.25\\0&1\end{bmatrix},
\end{equation*}
\begin{equation*}
    u(t) = -Kx(t),
\end{equation*}
\begin{equation*}
   K = \begin{bmatrix}1.9&-7.5\\
    1&7\end{bmatrix}.
\end{equation*}
The constant gain matrix $K$ can be any matrix such that the closed loop system is stable. 

\subsection{Node Placement and Oracle Calls}\label{section:resultsmain}
In this section, we highlight the advantages and disadvantages of q-RRT and Pq-RRT over RRT and Parallel RRT in being able to add nodes to the tree. We show performance compared to runtime in seconds, which we call wall-clock time to highlight that this is the ``real" runtime of the simulations, and then to number of oracle calls, which functions as the projected runtime improvement if algorithms are run on quantum devices. For performance metrics, we consider two quantities: number of oracle calls and number of nodes. Number of oracle calls is the metric for comparing how much Parallel q-RRT is able to speed up computation as compared to q-RRT. For all other comparisons, number of nodes is the chosen metric, as it signifies on a functional level the ability of each algorithm to search for a solution. Each point in Fig.~\ref{fig:oraclespeed}-~\ref{fig:nodespeedsimulated} represents one $30$-node tree creation, chosen to showcase average performance.

First, in Fig.~\ref{fig:oraclespeed} we show the wall-clock speed of Pq-RRT and q-RRT in being able to perform oracle calls to analyze the amount of computation speedup achieved. Next, we compare the wall clock speed of the two classical algorithms (in Fig.~\ref{fig:nodespeed}) and the two quantum algorithms (in Fig.~\ref{fig:nodespeedquantum}) to study the relative performance gain in parallelizing the quantum routines compared to classical parallel advantage. Lastly, in Fig.~\ref{fig:nodespeedsimulated}, we change to an oracle call (our quantum time surrogate) vs node creation comparison to show the efficiency of the quantum algorithms in admitting nodes to the tree. The performance advantage of Pq-RRT is inferred to combine the advantage shown over q-RRT in Fig.~\ref{fig:nodespeedquantum} and the advantage shown over the classical algorithms shown in Fig.~\ref{fig:nodespeedsimulated}.

\begin{figure}[h]
	\centering
	\includegraphics[width=.48\textwidth]{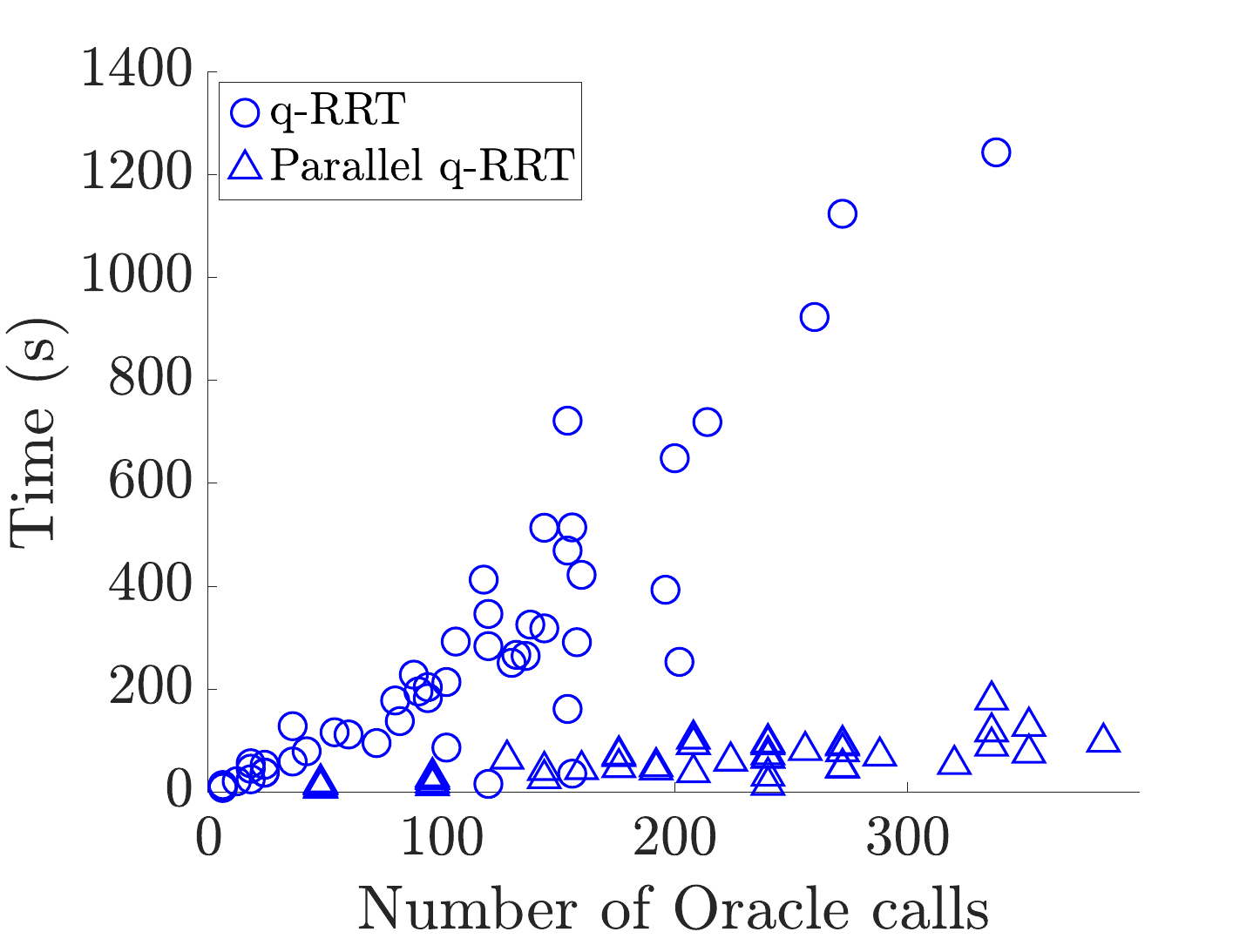}
	\caption{Comparison of the wall-clock speed (in seconds) of q-RRT, and shared database Pq-RRT in performing oracle calls or reachability tests.}
	\label{fig:oraclespeed}
\end{figure}
Figure~\ref{fig:oraclespeed} depicts the wall-clock speed of q-RRT and parallel q-RRT in performing oracle calls. Pq-RRT performs oracle calls in less time than q-RRT, as it more efficiently uses multiple workers to retrieve information from created databases. The relationship between oracle calls and time is approximately linear (as expected), and with a linear fit (using linear least squares), Pq-RRT shows a smaller slope ($0.20$) compared to q-RRT ($3.52$), with slope referring to seconds per oracle call (lower is more efficient). A classical computing shortcut is used which allows Pq-RRT to be $17.6$ times more efficient in performing work. In quantum computing simulation, a single created database is analyzed for reachability once, then amplified once, and each worker then measures a solution. This shortcut would not be possible on a quantum device, as qubits, once they are created, cannot be copied, so each worker would need to perform the reachability analysis separately. This would change the expected slope difference to be approximately $8$ times less than $17.6$ (as $8$ cores are used), for a total work (oracle call) efficiency gain of $2.2$.

\begin{figure}[h]
	\centering
	\includegraphics[width=.48\textwidth]{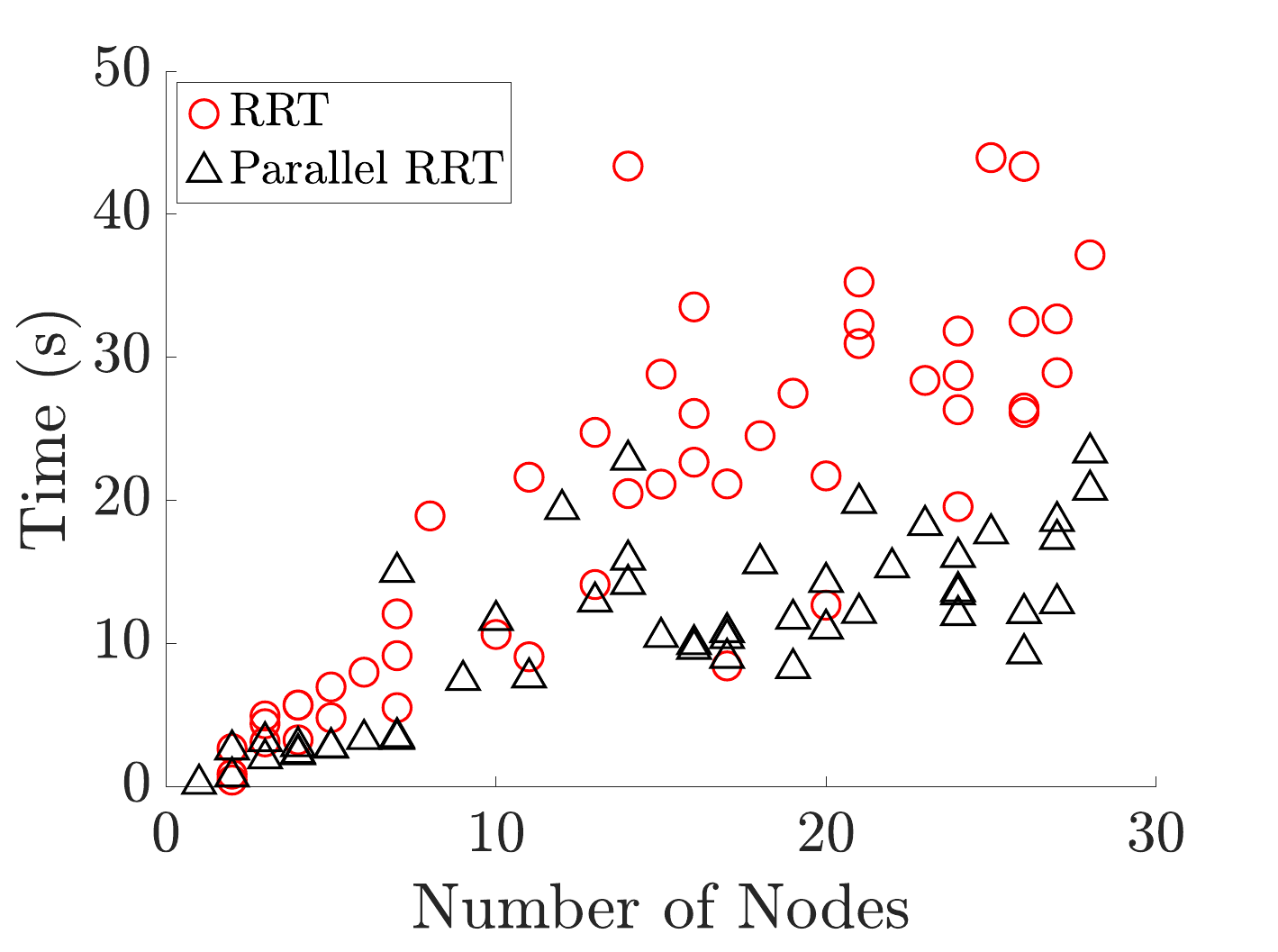}
	\caption{Comparison of the wall-clock speed (in seconds) of RRT and Parallel RRT in admitting reachable states to the tree.}
	\label{fig:nodespeed}
\end{figure}

\begin{figure}[h]
	\centering
	\includegraphics[width=.48\textwidth]{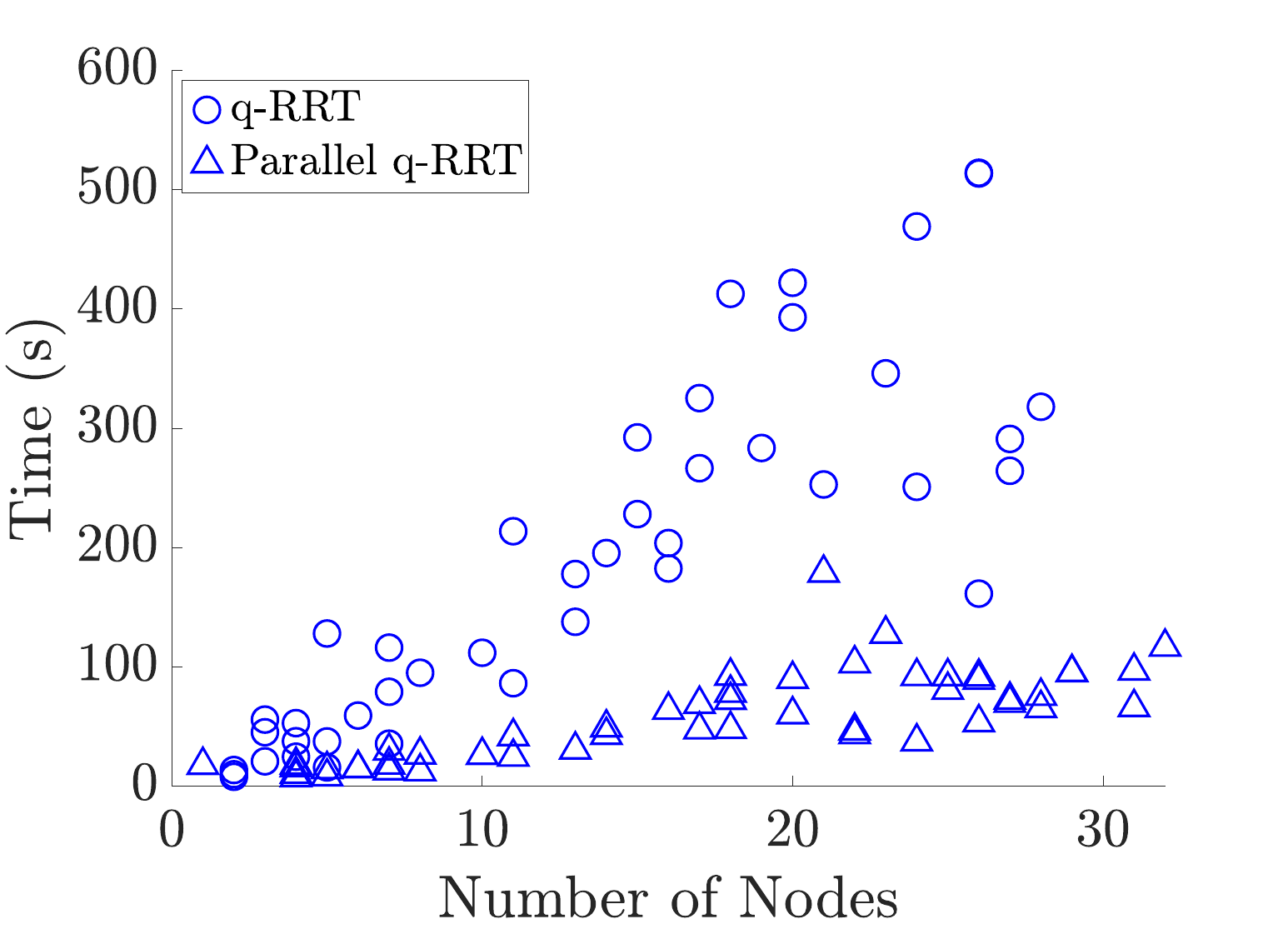}
	\caption{Comparison of the wall-clock speed (in seconds) of q-RRT and Pq-RRT in admitting reachable states to the tree.}
	\label{fig:nodespeedquantum}
\end{figure}

Figure~\ref{fig:nodespeed} depicts the wall-clock speed of RRT and Parallel RRT in admitting nodes to the graph, as opposed to performing oracle calls. The same data for q-RRT and Pq-RRT is shown in Fig.~\ref{fig:nodespeedquantum}. This comparison factors in differing node-admission oracle call efficiencies. The parallel versions of both algorithms, Parallel RRT and Pq-RRT, each are on average more time efficient than the non-parallel versions in admitting reachable states to the tree. 

The intuition behind slope in Figures~\ref{fig:nodespeed} and~\ref{fig:nodespeedquantum} is seconds per node, with lower numbers meaning more efficient.
Pq-RRT (slope $3.17$) in particular shows greater improvement over q-RRT (slope $25.3$) than Parallel RRT (slope $0.58$) over RRT (slope $1.23$), as is evidenced by a larger difference in slope ($8.0$-fold efficiency increase compared to $2.1$-fold), as calculated with a linear fit and linear least squares. 
The quantum algorithms, when measured by real time, lag behind both non-quantum RRT versions because they are not benchmarked on quantum computers (see the $y$-axis label differences between Fig.~\ref{fig:nodespeed} and Fig.~\ref{fig:nodespeedquantum}). The quantum computing simulations are performed via large arrays on classical devices. On a quantum device, we expect the run-time to be analogous to oracle calls, as discussed next.

\begin{figure}[h]
	\centering
	\includegraphics[width=.48\textwidth]{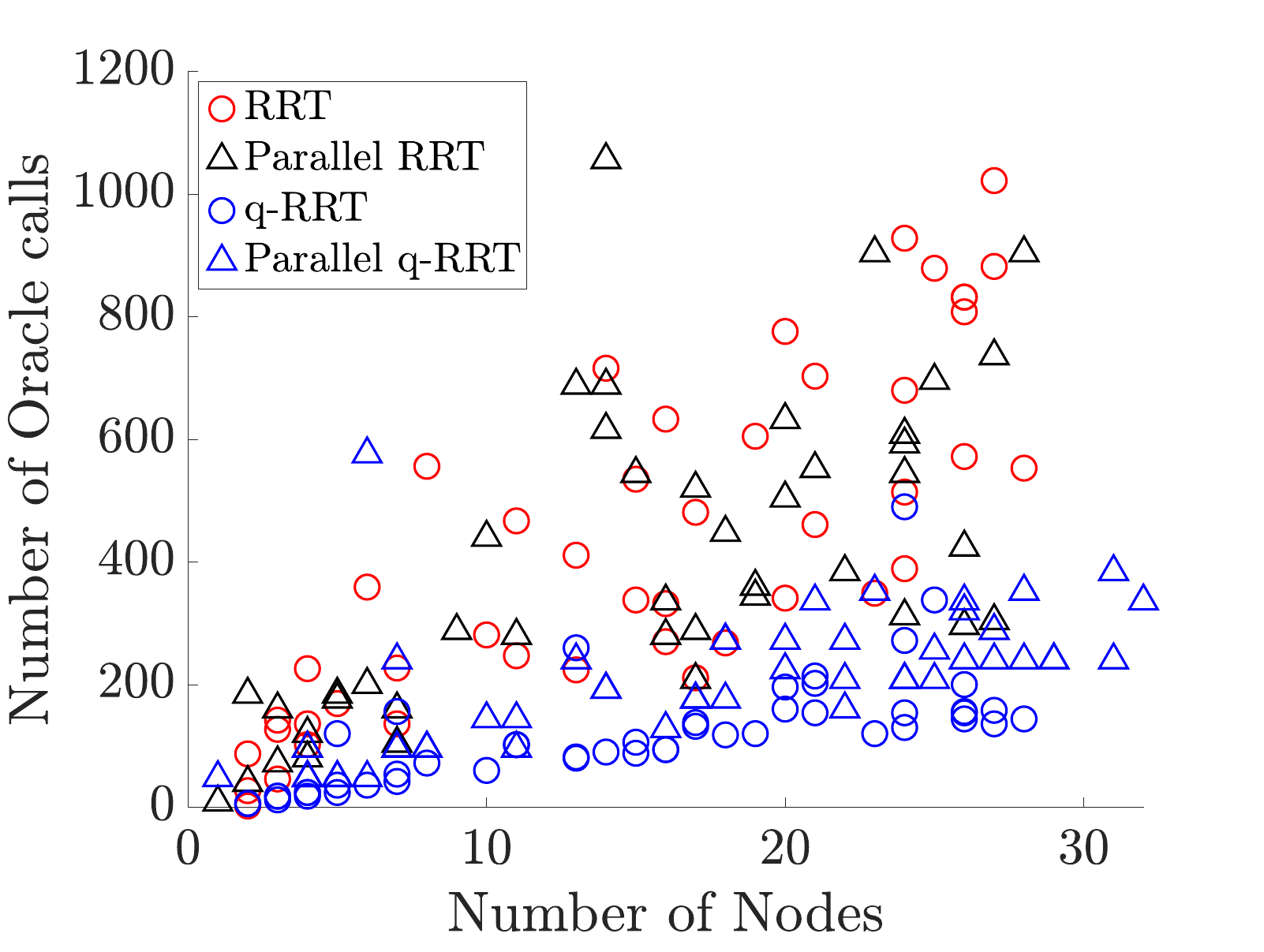}
	\caption{Comparison of the oracle call efficiency of RRT, Parallel RRT, q-RRT, and Pq-RRT in admitting reachable states to the tree.}
	\label{fig:nodespeedsimulated}
\end{figure}

Figure~\ref{fig:nodespeedsimulated} depicts the oracle call efficiency of all four algorithms in admitting reachable states to the tree as a function of the number of oracle calls it takes. This figure is analogous to expected run-time when the quantum algorithms are executed on a quantum device. Slopes are found with a linear fit using linear least squares and represent the number of oracle calls per node, with lower being more efficient. The efficiency advantage of q-RRT (slope $7.6$) and Pq-RRT (slope $10.7$) in admitting reachable states is shown over RRT (slope $21.5$) and Parallel RRT (slope $19.4$). The q-RRT algorithm is more efficient than Pq-RRT due to the fact that multiple workers can simultaneously return the same solution from a database (as explored in Props.~\ref{lemma:1} to~\ref{lemma:6}), and repeat solutions are discarded. However, Pq-RRT is capable of making simultaneous oracle calls with different workers, so for parallel vs not parallel time comparisons we refer the reader to Fig.~\ref{fig:nodespeed} and Fig.~\ref{fig:nodespeedquantum}.

The conclusions of the above analysis are the following: Pq-RRT is more time-efficient than q-RRT in performing work and placing nodes, Pq-RRT shows a greater time efficiency increase over q-RRT than Parallel RRT does over RRT, and q-RRT is slightly more oracle-call-efficient than Pq-RRT, but both quantum algorithms are more oracle-call-efficient than the classical algorithms. 

\subsection{Exploration Speed}\label{section:resultsspeed}
The results of this section are extensions to our results in~\cite{PL-BB-SM:23} between q-RRT and RRT, showing q-RRT's ability to explore quickly and in a generalized environment. We show a heat-map of state space nodes placed within a certain number of oracle calls. Oracle calls are chosen as a substitute to time because the quantum computer simulation performs slowly on classical devices. Actual runtime is expected to be analogous to the number of oracle calls, as reachability tests for the local planner consume the majority of the algorithm runtimes. Each algorithm is tested over $100$ trials. Each trial is cut off after a certain number of oracle calls to show each algorithm's speed of node placement. Let oracle efficiency be the ratio of total nodes placed over total oracle calls. In each figure, the red circle refers to a goal zone.

\begin{figure}[h]
	\centering
	\includegraphics[width=.4\textwidth]{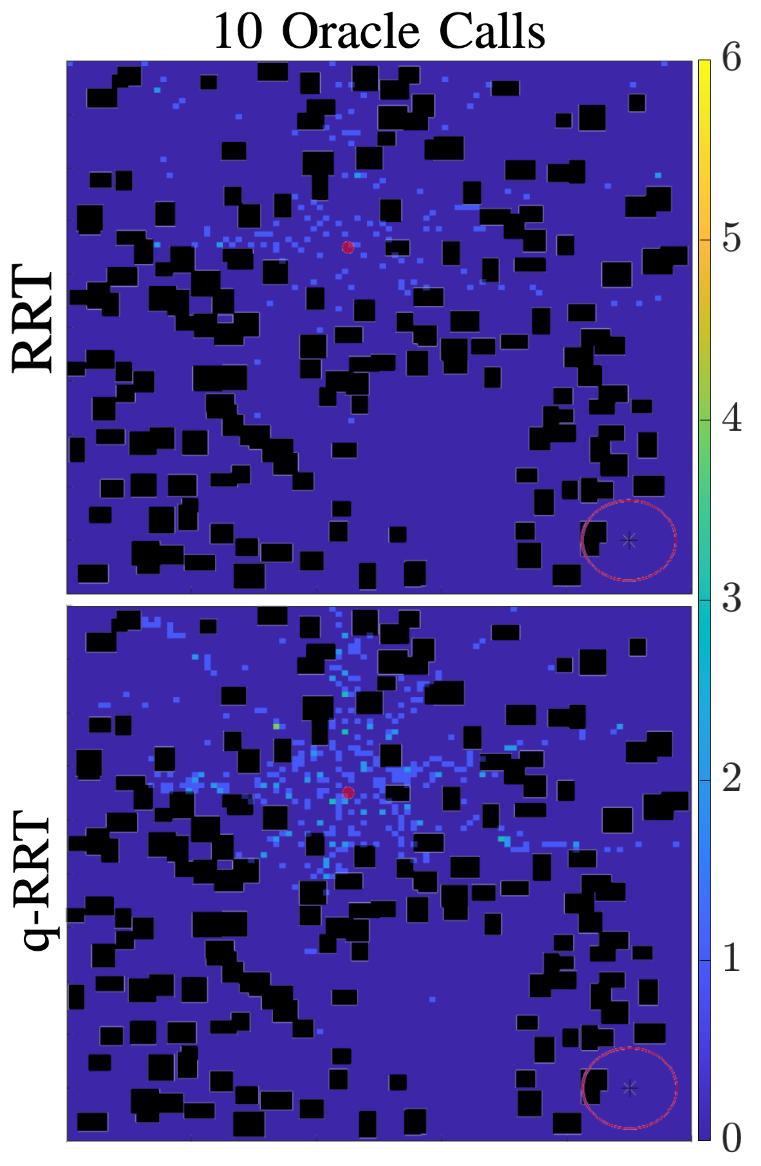}
	\caption{Comparison of initial exploration speeds (up to $10$ oracle calls) of RRT and q-RRT. Data is shown as a state space heat-map of node placements over $100$ trials of each algorithm in the shown obstacle environment. A goal zone is shown as a red ring in the bottom right, and the heatmap color key is shown on the right of the graph.}
	\label{fig:heatmap1}
\end{figure}
Figure~\ref{fig:heatmap1} depicts the initial exploration speeds, from $0$ to $10$ oracle calls, of RRT and q-RRT. Each path planning problem is cut off after $10$ oracle calls and a heatmap is created of the total node placement in the state space over $100$ trials. The q-RRT method shows much faster initial node placements over RRT, admitting $372$ nodes with an oracle efficiency of $31.2\%$. The q-RRT method has more than a thousand total oracle calls due to the inclusion of a finalizer line before nodes are admitted to the tree. RRT admitted $125$ nodes with an oracle efficiency of $12.5\%$. Node placement is more dense both in the initial node pocket and along lines exploring outward between obstacles away from the initial node.

\begin{figure}[h]
	\centering
	\includegraphics[width=.4\textwidth]{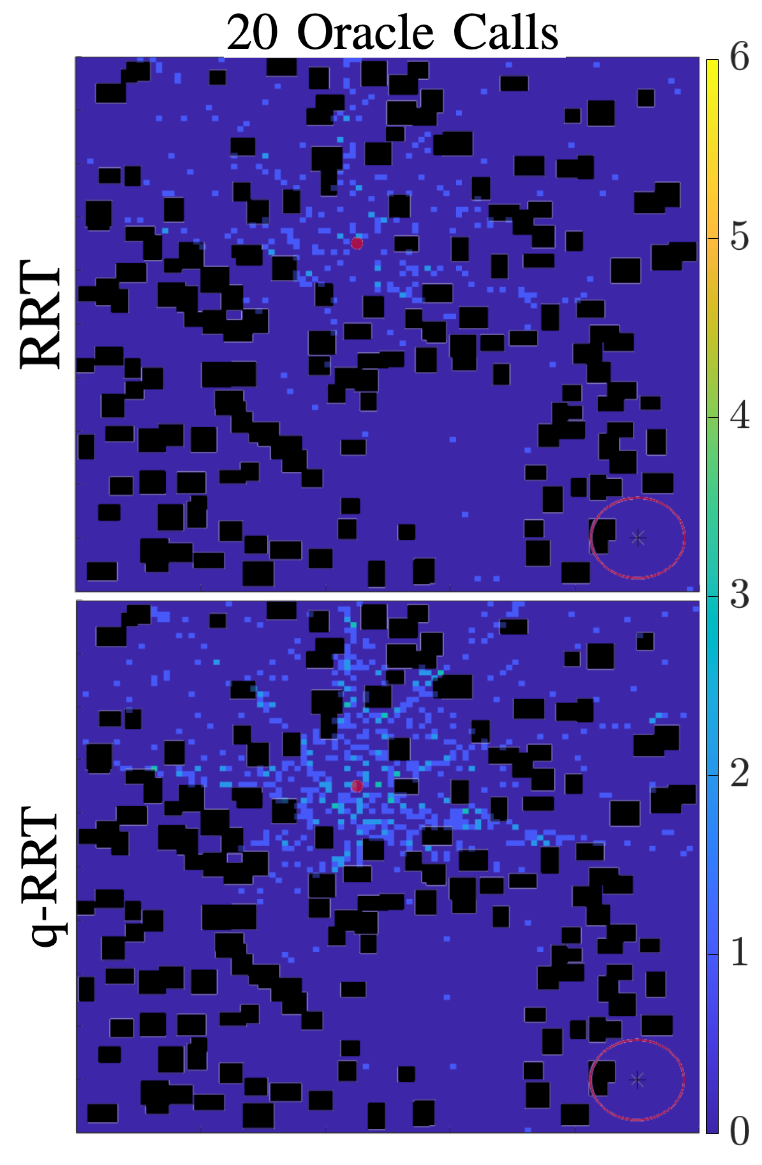}
	\caption{Comparison of middle-time exploration speeds (up to $20$ oracle calls) of RRT and q-RRT. Data is shown as a state space heat-map of node placements over $100$ trials of each algorithm in the shown obstacle environment.}
	\label{fig:heatmap2}
\end{figure}
Figure~\ref{fig:heatmap2} depicts the middle-time exploration speed, from $0$ to $20$ oracle calls, of RRT and q-RRT. Similarly, each path planning problem is cut off after $20$ oracle calls and a heatmap created from the total node placement of each algorithm over $100$ trials. The q-RRT method shows much faster and more full middle-time node placement, admitting $650$ nodes with an oracle efficiency of $31.0\%$. RRT admitted $231$ nodes with an oracle efficiency of $11.6\%$. Node placement is more ``full" in the initial pcket, and is much more dense along lines exploring out between obstacles from the initial node.

\begin{figure}[h]
	\centering
	\includegraphics[width=.4\textwidth]{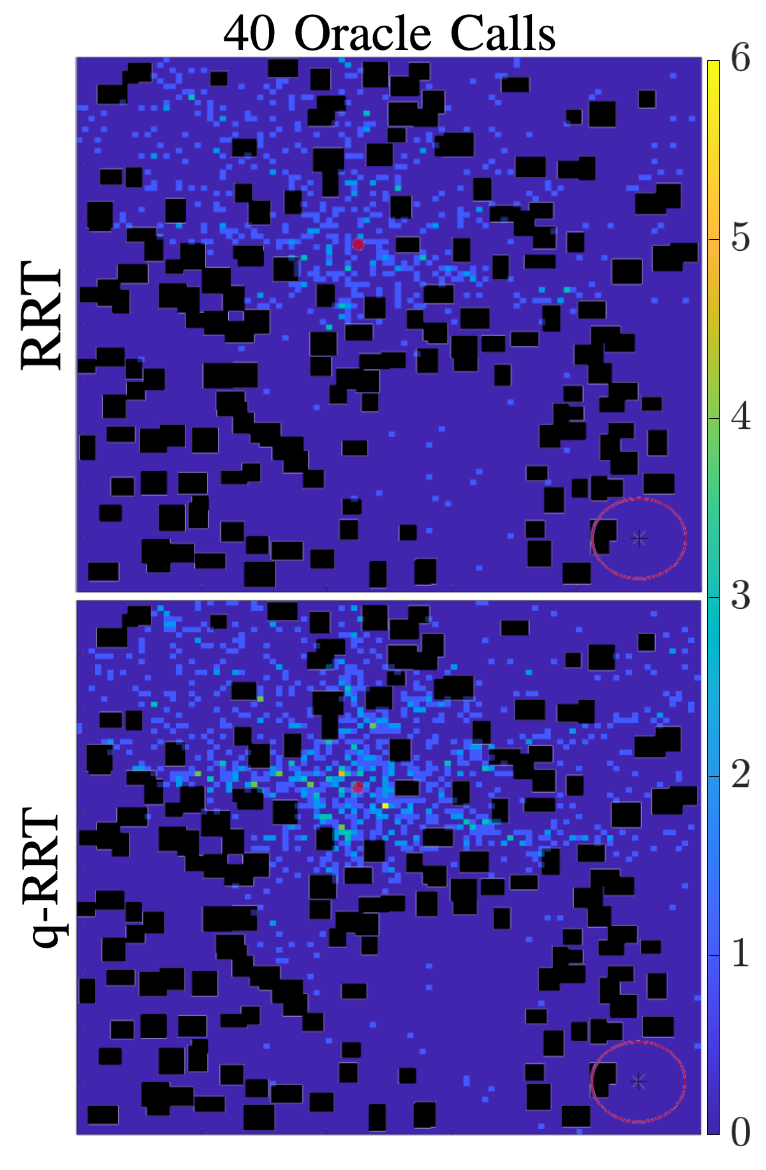}
	\caption{Comparison of late-time exploration speeds (up to $60$ oracle calls) of RRT and q-RRT. Data is shown as a state space heat-map of node placements over $100$ trials of each algorithm in the shown obstacle environment.}
	\label{fig:heatmap3}
\end{figure}
Figure~\ref{fig:heatmap3} depicts the ``late-time'' exploration speed, from $0$ to $40$ oracle calls, of RRT and q-RRT. Heatmap creation is again similar to previous. The q-RRT method has admitted more nodes in nearby pockets, and has a more dense spread of nodes in further away regions, admitting $1091$ nodes with an oracle efficiency of $26.0\%$, compared to RRT, which admitted $526$ nodes with an oracle efficiency of $13.2\%$. The average oracle efficiency of q-RRT has dropped somewhat compared to the initial and middle time exploration, and this is due to the fact that, as the existing tree grows, new random points are more likely to be reachable to the existing graph. This serves to allow RRT to catch up in terms of efficiency, and q-RRT's created database, on average, has allowed more solutions. The quantum version of the algorithms thrive (in comparison) in situations where there are few solutions. Heatmaps for Pq-RRT were also created for $10$, $20$, and $40$ oracle call cases as shown in Fig.~\ref{fig:heatmap1}-~\ref{fig:heatmap3}. The Pq-RRT heatmaps were omitted, as results were largely similar between q-RRT and Pq-RRT. This is similar to findings in Fig.~\ref{fig:nodespeedsimulated} that, when compared over oracle calls (quantum time surrogate), Pq-RRT's advantage is not apparent, as Pq-RRT is able to make simultaneous oracle calls. 

\subsection{Narrow Corridor Exploration}\label{section:resultsthinchannel}
The ability of motion planning algorithms to find paths through narrow corridor environments serves as a benchmark for the ability to find difficult solutions in narrow spaces.
In Figs.~\ref{fig:thinchannelqRRT} and~\ref{fig:thinchannelRRT} we show, through a heatmap, the ability of q-RRT to find passage through a narrow corridor when compared to RRT. The figures depict a heatmap of node placements of $50$ trials of each algorithm in the overlaid environment, where each method is cut off after $25$ oracle calls to analyze ability to quickly place nodes in the narrow corridor. Obstacles are depicted in black, and are distributed randomly on both sides of the narrow corridor, which is created by $2$ large obstacles. 

The q-RRT method placed $47$ nodes in the narrow corridor, compared to RRT's $14$ nodes. Results are presented with no guided sampling or known-goal direction to guide sampling. The q-RRT algorithm is quicker to find paths into narrow corridors toward possible goal locations.
\begin{figure}[h]
	\centering
	\includegraphics[width=.35\textwidth]{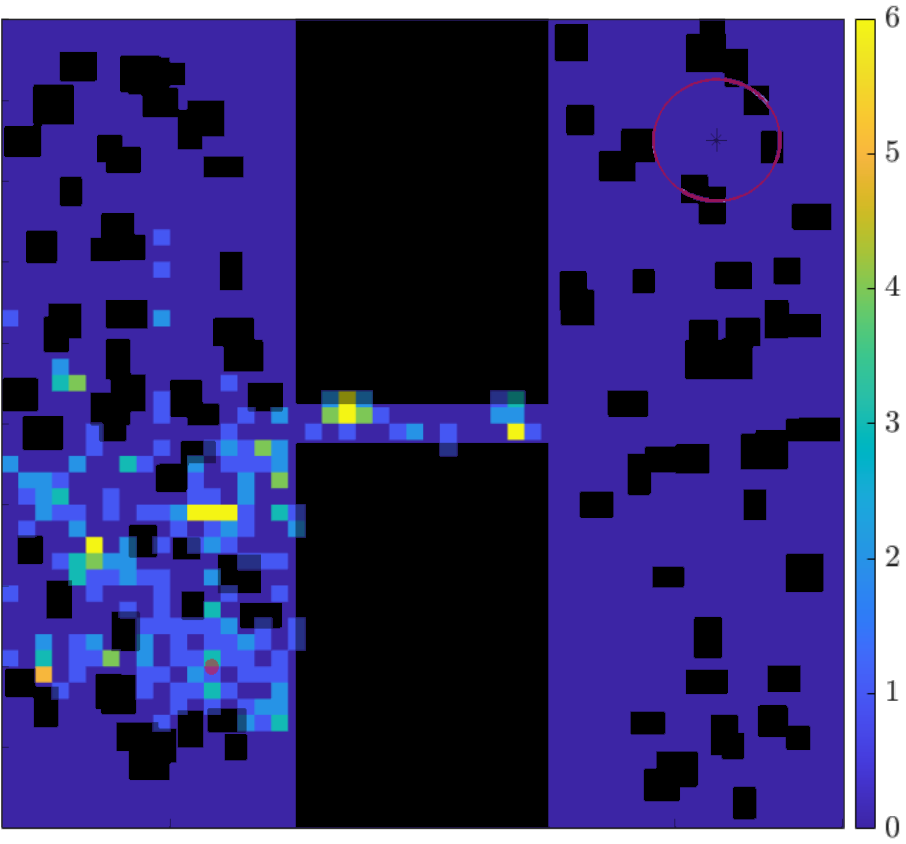}
	\caption{A heatmap of q-RRT's node placement in a narrow corridor environment (up to $25$ oracle calls) over $50$ trial runs, with $47$ nodes in the channel. Obstacles are depicted in black, and a color key of node placements is shown to the right of the graph.}
	\label{fig:thinchannelqRRT}
\end{figure}
\begin{figure}[h]
	\centering
	\includegraphics[width=.35\textwidth]{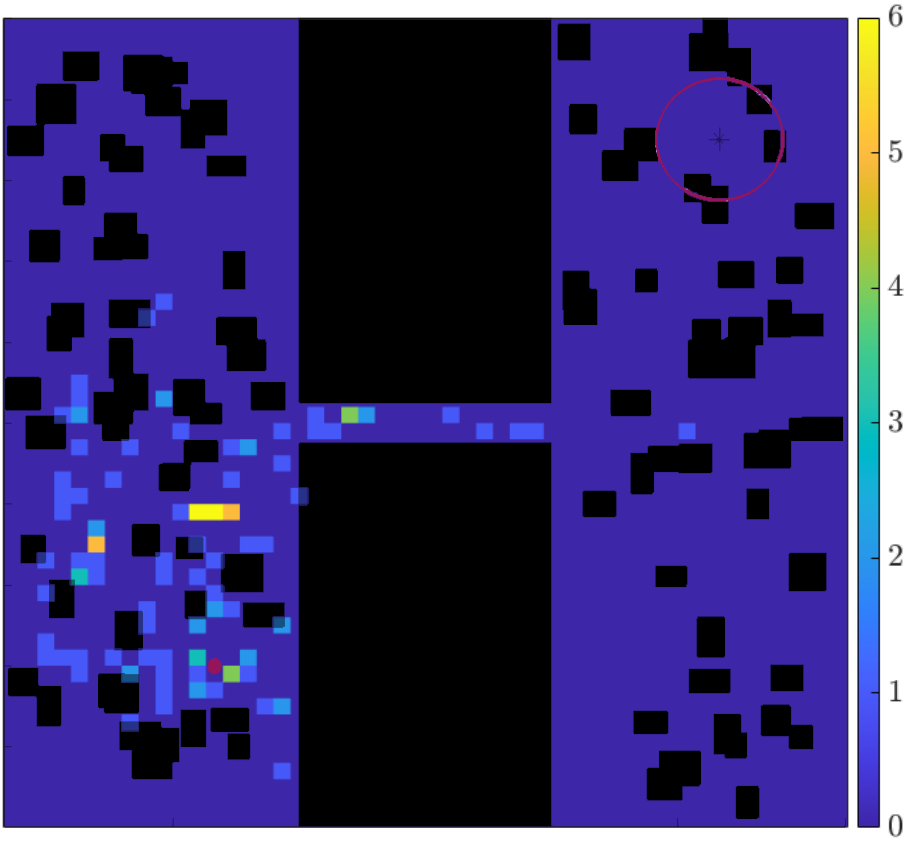}
	\caption{A heatmap of RRT's node placement in a narrow corridor environment (up to $25$ oracle calls) over $50$ trial runs, with $14$ nodes in the channel.}
	\label{fig:thinchannelRRT}
\end{figure}
\subsection{Quantum Database Annealing}\label{section:resultsannealing}
We compare the abilities of Quantum Database Annealing and standard q-RRT database construction to create trees that spread across larger configuration spaces with a large number ($6025$) of obstacles. In this formulation, Quantum Database Annealing is initially creating databases of points at a distance between $2.7$ and $4.2$ units from current nodes, then dropping that range to between $0.8$ and $2.0$ units to fill in the space around the initially spread tree. On the other hand, q-RRT with standard database construction is sampling across the entire configuration space $C$. Both algorithms are using databases of size $2^9$. Fig.~\ref{fig:annealingtree} depicts a $16$-node tree with initial fast expansion made with Quantum Database Annealing, and Fig.~\ref{fig:annealingtreelowtemp} shows continued node addition to a $48$-node tree with lower temperature to fill in the area around the initial spread tree. Fig.~\ref{fig:qrrttree} depicts the standard q-RRT created $16$-node tree in the same environment, to compare against Fig.~\ref{fig:annealingtree}. Obstacles are depicted as small black rectangles, the root node of each tree is shown as a black circle, nodes in each tree are shown as red circles, and parent child connections are shown as black lines.

The resulting trees differ in how spread they are for the same number of oracle calls (the quantum analog of runtime). Quantum Database Annealing initially creates nodes an average of $3.68$ units away (with the above temperature setting) from their parent and standard q-RRT creates nodes an average of $1.70$ units away from their parent. For a fixed number and size of obstacles, it should be noted that sampling parameters affect the average distance in the q-RRT tree, and different average distances can be obtained by varying the size of the database. For uniform sampling over $C$, as the database becomes larger, the average distance drops, as nodes are more frequently found near existing nodes. For equal-sized large databases, in the same amount of (quantum) time, QDA is able to create trees with more spread, as only further away nodes are admitted to the database. The temperature construct allows a balance between exploration and density of nodes, enabling a version of q-RRT that can connect distant regions of a configuration space very quickly before back-filling with lower temperature. 

\begin{figure}[h]
	\centering
	\includegraphics[width=.48\textwidth]{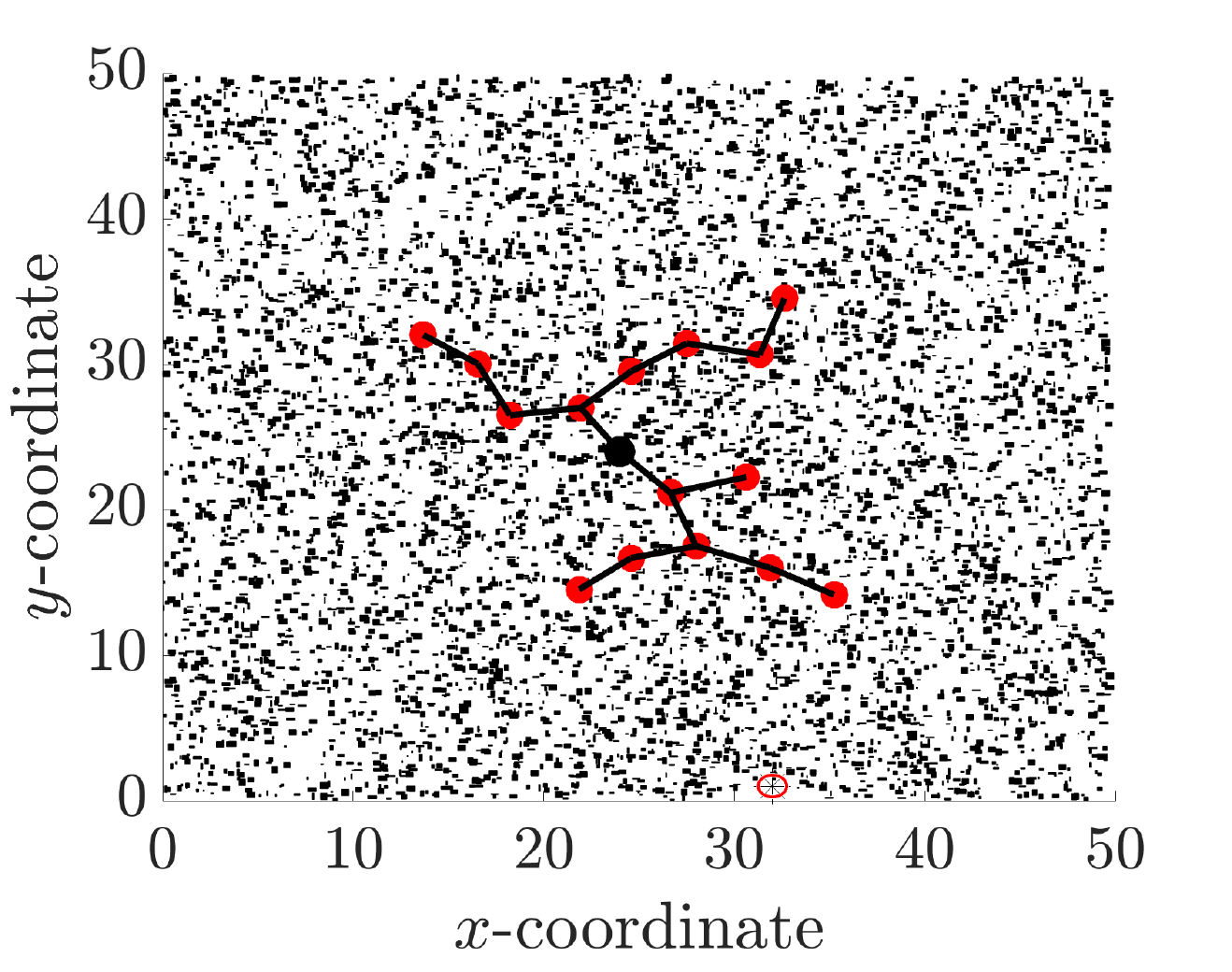}
	\caption{A $16$ node tree created with Quantum Database Annealing with high temperature, showing fast initial exploration. The root node is a black circle and tree nodes are red circles. Parent-child relationships are shown via black lines, and $6025$ obstacles are depicted as small black rectangles.}
	\label{fig:annealingtree}
\end{figure}

\begin{figure}[h]
	\centering
	\includegraphics[width=.48\textwidth]{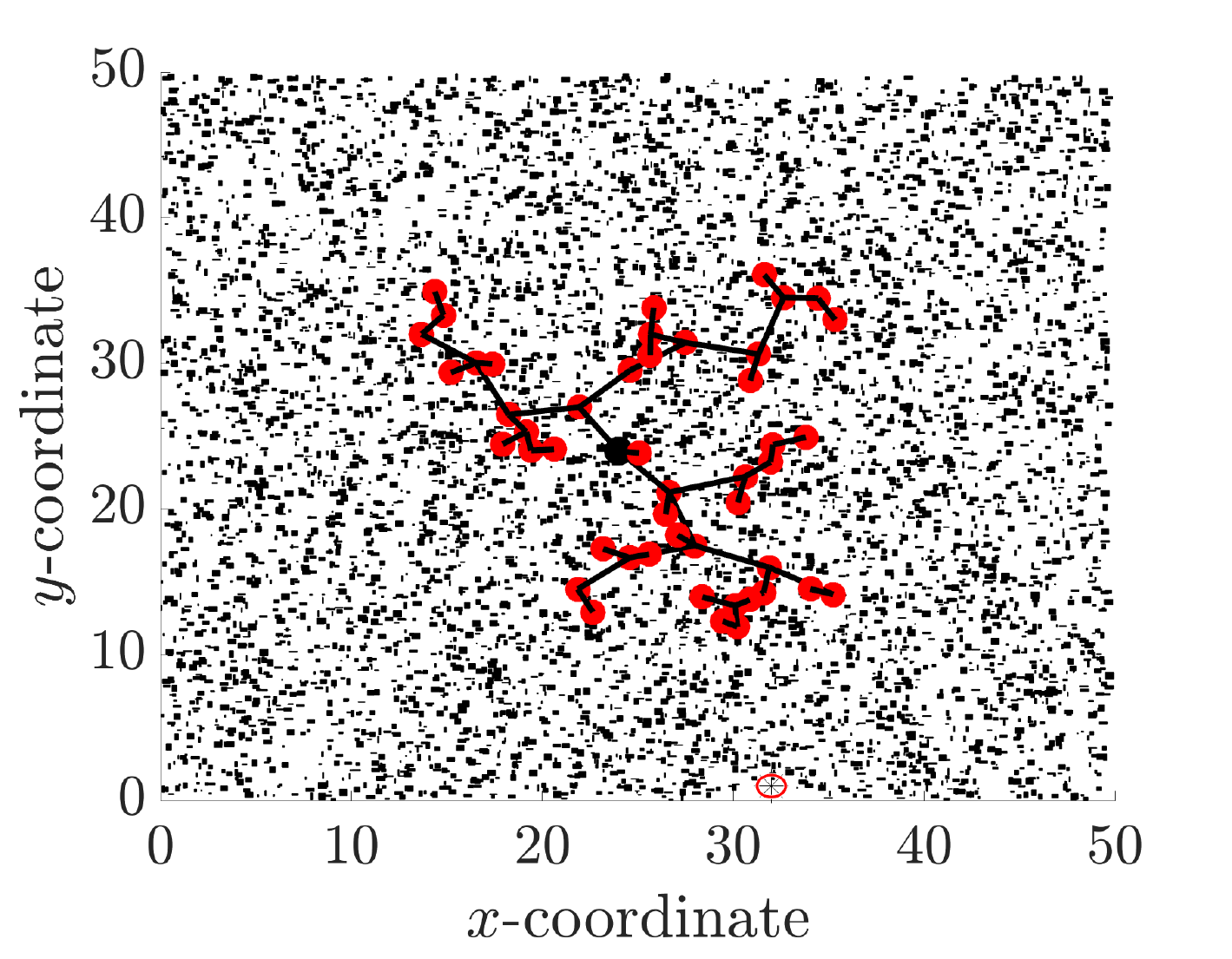}
	\caption{A $48$ node tree created with Quantum Database Annealing with initial high temperature, then a dropping temperature, showing how temperature can be used to fill in the area around a spread tree. The root node is a black circle and tree nodes are red circles. Parent-child relationships are shown via black lines, and $6025$ obstacles are depicted as small black rectangles.}
	\label{fig:annealingtreelowtemp}
\end{figure}

\begin{figure}[h]
	\centering
	\includegraphics[width=.48\textwidth]{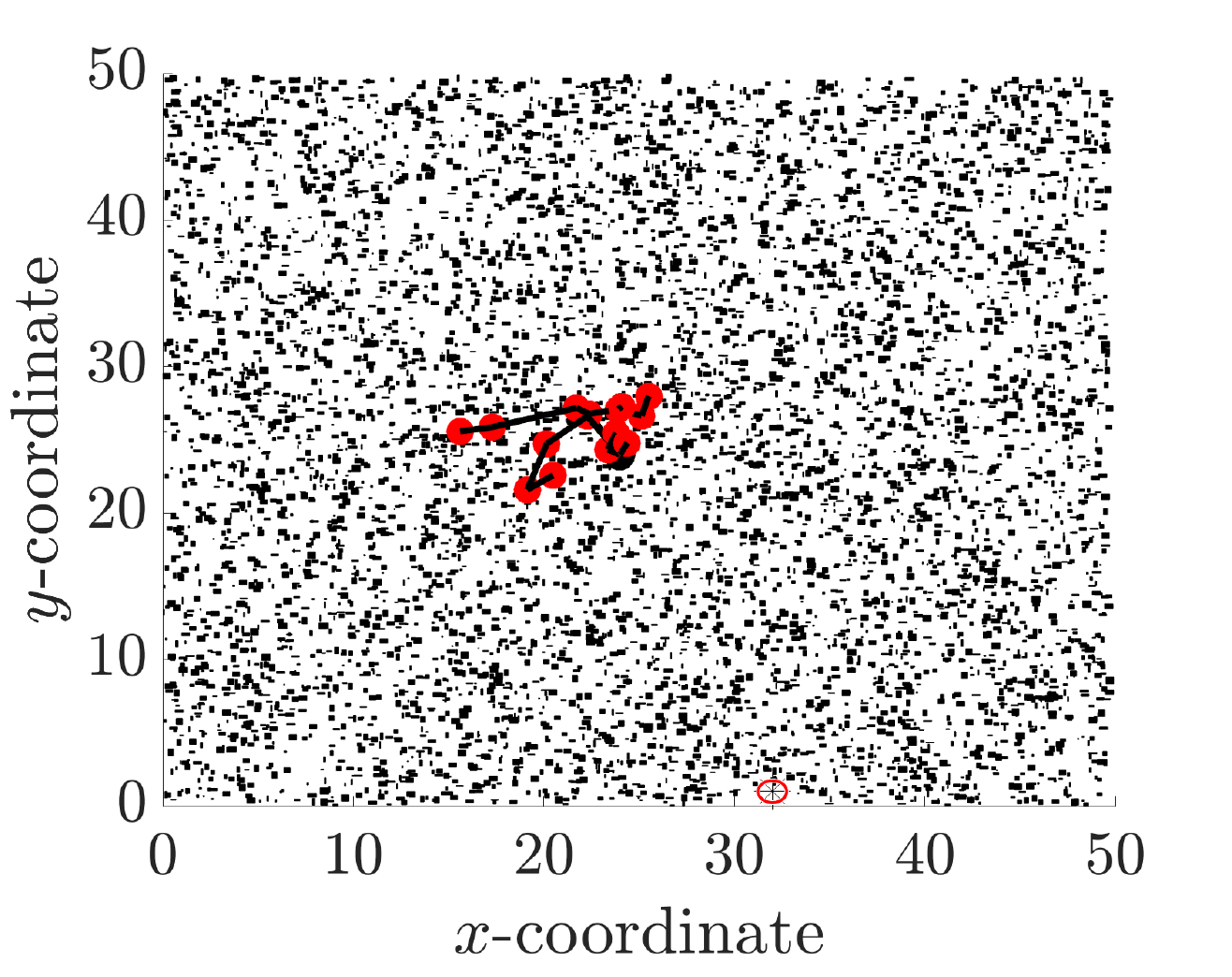}
	\caption{A $16$ node tree created with q-RRT (with standard database construction) in the same $6025$ obstacles environment.}
	\label{fig:qrrttree}
\end{figure}

\section{Conclusion}
To generalize and extend q-RRT, we provide analysis in more general obstacle environments, a formulation of q-RRT with parallel quantum computers, and a database building strategy based on simulated annealing. The Parallel Quantum RRT algorithm uses parallel quantum computers in a manager-worker formulation to provide simultaneous measurements of a shared database, allowing more time-efficient tree construction with a higher exploration speed. We also provide key probability results for parallel quantum computers searching the same database in order to predict parallel architecture efficiency. Quantum Database Annealing uses a temperature construct to guide database construction, providing trees that initially spread more quickly compared to those created with standard database construction, followed by back-fill behavior at lower temperatures. To support these claims, we provide analysis in the form of efficiency and run-time results, heatmaps for speed-of-exploration results, narrow corridor environment results, and database construction comparisons. Future work includes expanding on alternate methods of database construction and creating path planning algorithms that rely on alternate quantum algorithms to QAA.

\bibliographystyle{IEEEtran}
\bibliography{alias,SM,SMD-add,JC}

\begin{IEEEbiography}[{\includegraphics[width=1in,height=1.25in,clip,keepaspectratio]{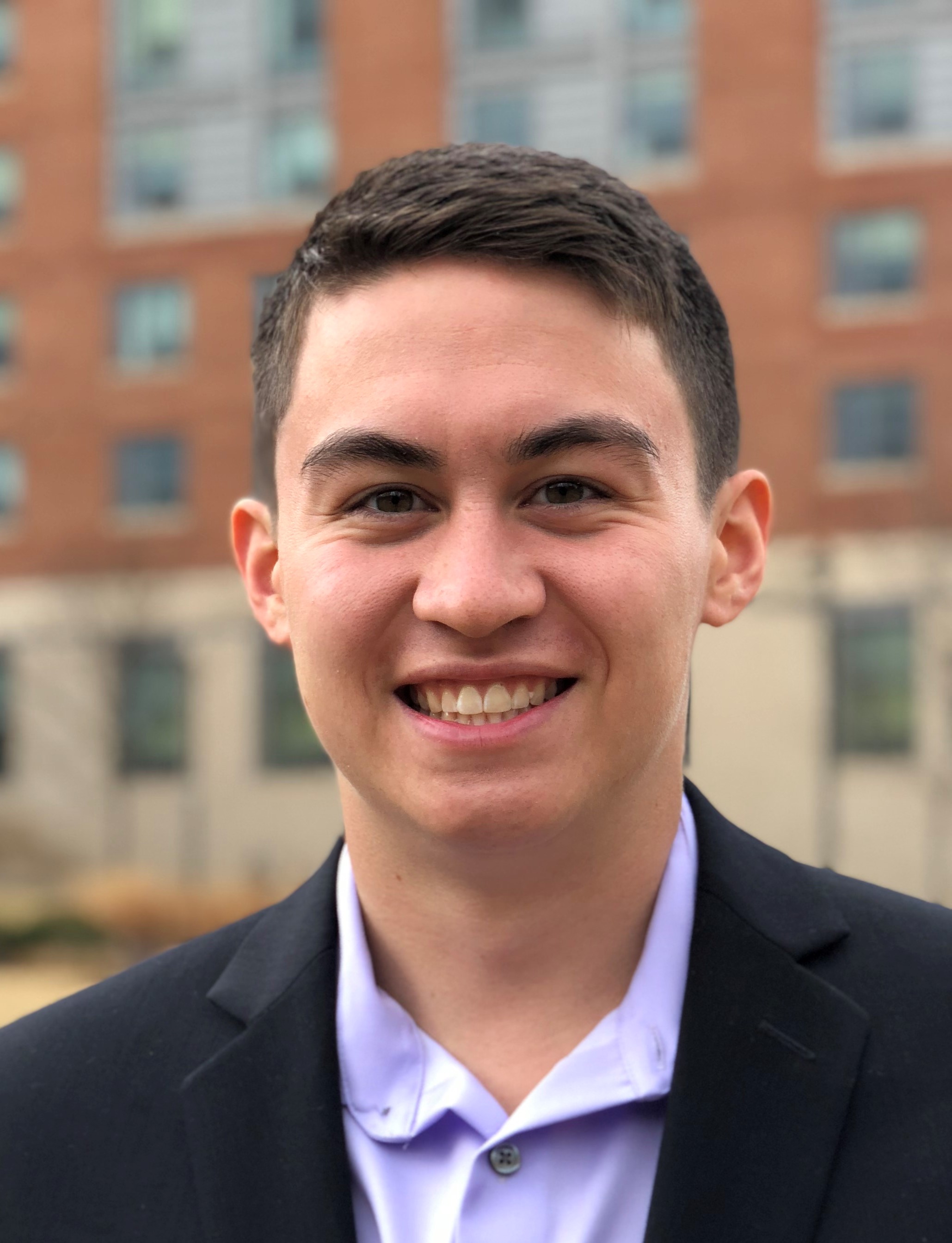}}]{Paul Lathrop} received the B.S. degree in Aerospace Engineering from the University of Maryland, College Park in 2019, the M.S degree in Aerospace Engineering from the University of California, San Diego in 2021, and is currently a Ph.D. candidate in Mechanical and Aerospace Engineering at the University of California, San Diego. His research interests include safety and uncertainty in robotic motion planning algorithms and quantum computing.
\end{IEEEbiography}
\begin{IEEEbiography}[{\includegraphics[width=1in,height=1.25in,clip,keepaspectratio]{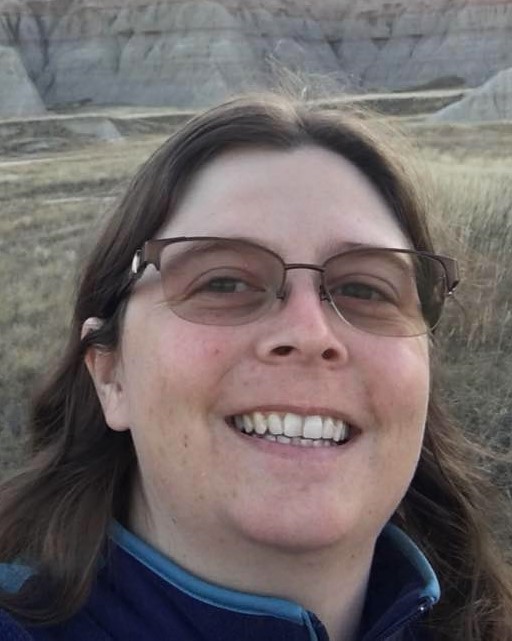}}]{Beth Boardman}
received a Ph.D. degree in aerospace engineering from the University of California, San Diego, USA in 2017 and a B.S. and M.S. in aeronautics and astronautics from the University of Washington, Seattle, USA in 2010 and 2012, respectively. She has worked as a Research and Development Engineer at Los Alamos National Laboratory since 2018. Her research interests include robotics and automation. Beth is currently the Team Leader for the Innovative Robotics team in the Automation and Control group. She is also the LANL Robotics and Automation Summer School program leader.
\end{IEEEbiography}
\begin{IEEEbiography}[{\includegraphics[width=1in,height=1.25in,clip,keepaspectratio]{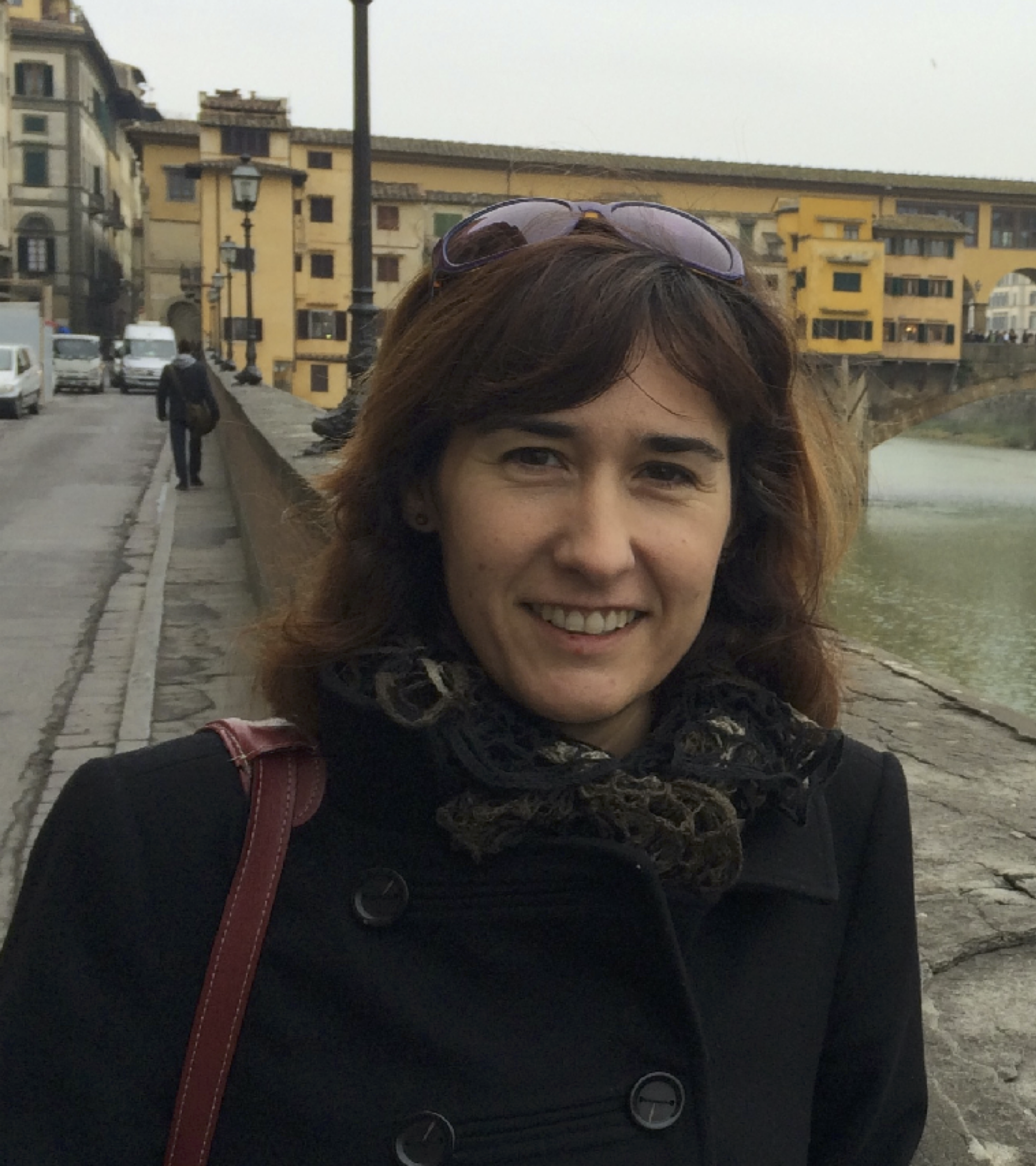}}]{Sonia Mart\'{i}nez}
 (M'02-SM'07-F'18) is a Professor of Mechanical and
  Aerospace Engineering at the University of California, San Diego,
  CA, USA. She received her Ph.D. degree in Engineering Mathematics
  from the Universidad Carlos III de Madrid, Spain, in May 2002. She
  was a Visiting Assistant Professor of Applied Mathematics at the
  Technical University of Catalonia, Spain (2002-2003), a Postdoctoral
  Fulbright Fellow at the Coordinated Science Laboratory of the
  University of Illinois, Urbana-Champaign (2003-2004) and the Center
  for Control, Dynamical systems and Computation of the University of
  California, Santa Barbara (2004-2005).  Her research interests
  include the control of networked systems, multi-agent systems,
  nonlinear control theory, and planning algorithms in robotics. She
  is a Fellow of IEEE. She is a co-author (together with F. Bullo and
  J. Cort\'es) of ``Distributed Control of Robotic Networks''
  (Princeton University Press, 2009). She is a co-author (together
  with M. Zhu) of ``Distributed Optimization-based Control of
  Multi-agent Networks in Complex Environments'' (Springer, 2015).
  She is the Editor in Chief of the recently launched \textit{CSS IEEE Open
  Journal of Control Systems.}
\end{IEEEbiography}
\EOD
\end{document}